%% file: ClementGenitrini.tex
\documentclass[a4paper, UKenglish, cleveref, autoref, thm-restate]{lipics-v2021}

\title{An iterative approach for counting reduced ordered binary decision diagrams}

\titlerunning{An iterative approach for counting ROBDDs} %TODO optional, please use if title is longer than one line

\author{Julien Cl\'ement}{Normandie Universit\'e, UNICAEN, ENSICAEN, CNRS, GREYC - UMR 6072}
{Julien.Clement@unicaen.fr}{}{}
\author{Antoine Genitrini}{Sorbonne Universit\'e, CNRS, LIP6 - UMR 7606, F-75005 Paris, France }{Antoine.Genitrini@lip6.fr}{}{}

\authorrunning{J. Cl\'ement and A. Genitrini}

\Copyright{Julien Cl\'ement and Antoine Genitrini} %TODO mandatory, please use full first names. LIPIcs license is "CC-BY";  http://creativecommons.org/licenses/by/3.0/

\ccsdesc[500]{Mathematics of computing~Combinatorial algorithms}
\ccsdesc[500]{Mathematics of computing~Generating functions}
\ccsdesc[500]{Mathematics of computing~Combinatoric problems}
\ccsdesc[500]{Theory of computation~Generating random combinatorial structures}
\ccsdesc[500]{Information systems~Data compression}
\ccsdesc[500]{Theory of computation~Data compression}

\keywords{Boolean Function, Reduced Ordered Binary Decision Diagram (\textsc{robdd}), Enumerative Combinatorics, %Uniform Random sampling,
	Directed Acyclic Graph}%TODO mandatory; please add comma-separated list of keywords

\acknowledgements{The authors thank the anonymous referees for their comments and suggested improvements. All these remarks have increased the quality of the paper.}%optional

\nolinenumbers %uncomment to disable line numbering

\usepackage[utf8]{inputenc} % allow utf-8 input
\usepackage[T1]{fontenc}    % use 8-bit T1 fonts
\usepackage{hyperref}       % hyperlinks
\usepackage{url}            % simple URL typesetting
\usepackage{booktabs}       % professional-quality tables
\usepackage{amsfonts}       % blackboard math symbols
\usepackage{microtype}      % microtypography
\usepackage{tikz}
\usetikzlibrary{snakes,arrows,shapes}

\usepackage{amsmath,amsthm,amssymb}
\usepackage{stmaryrd}
\usepackage{mathtools}
\usepackage{caption}
\usepackage[ruled,vlined,plain]{algorithm2e}
\DontPrintSemicolon
\SetKw{Continue}{continue}
\SetKw{And}{and}
\SetKw{Not}{not}
\SetKw{Or}{or}
\SetKw{New}{new}
\SetKw{From}{from}
\SetKw{DownTo}{downto}
\SetKw{To}{to}
\SetKw{In}{in}
\SetKwComment{Comment}{$\triangleright$ }{}

\usepackage{graphicx}
\usepackage{doi}
\newtheorem*{example*}{Example}
\newtheorem*{fact*}{Fact}

\newcommand\B{\mathcal{B}}
\newcommand{\BDD}{\textsc{robdd}}

\newcommand{\BDDs}{\textsc{robdd}s}

\newcommand{\True}{{\normalfont\texttt{1}}}
\newcommand{\False}{{\normalfont\texttt{0}}}
\newcommand{\Tnode}{\top}
\newcommand{\Fnode}{\bot}

\DeclareRobustCommand{\stirlingII}{\genfrac\{\}{0pt}{}}

\newcommand{\SEQ}[1]{[#1]}

\DeclareMathOperator*{\degree}{deg}

\begin{document}
	
	\maketitle
	
\begin{abstract}
  For three decades \emph{binary decision diagrams}, a data structure efficiently representing Boolean functions, 
  have been widely used in many distinct contexts like model verification, machine learning, cryptography and also resolution of combinatorial problems.
  The most famous variant, called reduced ordered binary decision diagram (\textsc{robdd} for short), can be viewed as the result of a compaction procedure on the full decision tree.
  A useful property is that once an order over the Boolean variables is fixed, each Boolean function is represented by exactly one \textsc{robdd}.
  In this paper we aim at computing the \emph{exact distribution of the Boolean functions in~$k$~variables according to the \textsc{robdd} size}, where the \textsc{robdd} size is equal to the number of decision nodes of the underlying directed acyclic graph (\textsc{dag}) structure.
  Recall the number of Boolean functions with~$k$ variables is equal to~$2^{2^k}$, which is of double exponential growth with respect to the number of variables. 
  The maximal size of a \textsc{robdd} with~$k$ variables is~$M_k \approx 2^k / k$. 
  % and thus, the support of the \textsc{robdd} size distribution is also of length~$M_k$.
  Apart from the natural combinatorial explosion observed, another difficulty for computing the distribution according to size is to take into account dependencies within the \textsc{dag} structure of \textsc{robdd}s.
%  , making~$M_k$ a natural complexity unit for our problem.
In this paper, we develop the first polynomial algorithm to derive the distribution of Boolean functions over~$k$ variables  with respect to \textsc{robdd} size denoted by~$n$.
  The algorithm computes the (enumerative) generating function of \textsc{robdd}s with~$k$ variables up to size~$n$. 
  It performs~$O(k\; n^4)$ arithmetical operations on integers and
  necessitates storing~$O((k+n) n^2)$ integers with bit length~$O(n\log n)$.
  Our new approach relies on a decomposition of \textsc{robdd}s layer by layer and on an inclusion-exclusion argument.

%  As a by-product, we present an efficient polynomial unranking algorithm for \textsc{robdd}s,
%  which in turn yields a uniform random sampler over the set of \textsc{robdd}s of a given size or of a given profile.
%  This is a great improvement to the classical random sampler which is uniform over the set of \emph{all} Boolean functions in~$k$ variables.
%  Indeed, due to the Shannon effect, the uniform distribution over Boolean functions is heavily biased to extremely complex functions, 
%  with near maximal \textsc{robdd} size, thus preventing in practice to sample \textsc{robdd}s with smaller size.
\end{abstract}

%\newpage
%
%\tableofcontents

%\newpage

\section{Introduction}
\label{sec:intro}

Three decades ago a central data structure in computer science,
designed to represent Boolean functions, emerged under the name of Binary Decision Diagrams (or \textsc{bdd}s)~\cite{Bryant86}.
Their algorithmic paradigm gives great advantages:
it is based on a \emph{divide-and-conquer} approach combined
with a compaction process. 
Their benefits compared to other Boolean representations are so obvious that several dozens of \textsc{bdd}
variants have been developed in recent years.
In his monograph~\cite{Wegener00}, Wegener presents several ones 
like \textsc{robdd}s~\cite{Bryant92}, \textsc{okfbdd}s~\cite{DSTBP94}, 
\textsc{qobdd}s~\cite{Wegener94}, \textsc{zbdd}s~\cite{Minato93}, and others. 
While most of these data structures are used in the context of verification~\cite{Wegener00}, 
they also appear, for example, in the context of cryptography~\cite{KJFB06} or knowledge compilation \cite{DarwicheMarquis02}.
Also, the size of the structure, depending on the compaction of a decision tree,
allows improving classification in the context of machine learning~\cite{MBFV04}.  
Finally, some specific \textsc{bdd}s are relevant to strategies for the resolution of combinatorial problems,
cf.~\cite[vol. 4]{Knuth11}, like the classical satisfiability count problem.

The classical way to represent the different diagrams consists in their embedding as directed acyclic graphs (or \textsc{dag}s).
In the following we are interested in the original form of decision diagrams
that are \textsc{robdd}s, for \emph{Reduced Ordered Binary Decision Diagrams}.
One of their fundamental properties relies on the single, thus canonical, representative property
for each Boolean function (with a given order over the Boolean variables).
In his book~\cite{Knuth11} Knuth recalls and proves several combinatorial results
for \textsc{robdd}s. He is, for example, interested in the profile of a typical \textsc{robdd},
or in the way to combine two structures to represent a more complex Boolean function.
However, thirty years after the takeoff of \textsc{robdd}s, the study of the distribution of Boolean functions with respect to the size, defined as the number of decision nodes (see~Figure~\ref{fig:example}), of the \textsc{dag} structure is not totally understood.
The main problem 
is that no recursive characterization describing the structure of \textsc{robdd}s is known,
as opposed, for instance, to the recursive decomposition of binary trees which is the core
approach in their combinatorial studies (profile, width, depth).

\subparagraph*{Related work.}
 An important step in the comprehension of the distribution of the Boolean functions
according to their \textsc{robdd} size has been achieved by Wegener~\cite{Wegener94} and improved by Gröpl \emph{et al.}~\cite{GPS04}. These authors proved that almost all functions have the same \textsc{robdd} size up to a factor of~$1 + o(1)$ when the number of variables~$k$ tends to infinity, exhibiting the Shannon effect (strong or weak depending on the value of~$k$).
The strong (respectively weak) Shannon
effect states that almost all functions have the same \textsc{robdd} size as the largest \textsc{robdd}s up to a factor of~$1 + f(k)$ with~$f(k)=o(1)$ (resp.~$f(k)=\Omega(1)$) as~$k$ tends to infinity (see also \cite{BV05}).
The reader may find an illustration of this phenomenon in Figure~\ref{fig:results} with the plot of the exact distribution for~$k=13$ variables.
A consequence of these first analyses is that picking \emph{uniformly at random} a Boolean function
whose \textsc{robdd} is small is not an easy task,
although in practice \textsc{robdd}s are often not of exponential size (with respect to~$k$).

In~\cite{NV19}, the authors study, experimentally, numerically,
and theoretically, the size of \textsc{robdd}s when the number~$k$ of variables is increasing.
However, their main approach relies on an exhaustive enumeration of the decision trees of all Boolean functions,
that are in a second step compressed into \textsc{robdd}s. The doubly exponential growth of Boolean functions
over~$k$ variables, equal to~$2^{2^k}$, allows only to compute the first values for~$k= 1, \dots, 4$.
Then the authors extrapolate the distributions by sampling decision trees (uniformly at random).

Later in the paper~\cite{CG20} we obtain similar combinatorial results.
Using a new approach based on a partial recursive decomposition, we partition the \textsc{robdd}s according to their profile (which describes the number of nodes per level in the \textsc{dag}). Another key feature of \cite{CG20} is that we can restrict ourselves to a maximal size~$n$ for \textsc{robdd}s, as opposed to the exhaustive-oriented approach of \cite{NV19}. 
Although more efficient, still algorithms with this approach are bound to be at least of complexity~$\Omega(n\; k^{3/2 \cdot k^2 / \log k})$, while using a huge amount of extra memory.
However, after a lengthy computation we obtain the exact distributions of the size of \textsc{robdd}s up to~$k=9$, thus partitioning the set of~$2^{512}$ Boolean functions into \textsc{robdd}s of sizes ranging from~$0$ to~$141$.

\subparagraph*{Main results.}
In this paper, we describe an algorithm that calculates the exact distribution for \textsc{robdd}s up to size~$n$
in time complexity~$O(k\; n^4)$ using~$O((n+k)n^2)$ extra storage memory for integers. 
To the best of our knowledge, this is the first polynomial complexity
algorithm computing the distribution of the \textsc{robdd} size for Boolean functions.
Our combinatorial approach is based on an iteration process instead of a recursive approach.

We improve drastically on previous work and all the extrapolated results presented
in~\cite{NV19} for Boolean functions up to~$13$ variables are now fully and exactly described.
Using a personal computer, in a couple of minutes we obtain an exhaustive counting of the \textsc{robdd}s
representing functions over~$11$ variables. 
With a computer with several hundreds gigabytes of \textsc{ram} we compute the distribution over~$13$ variables in about a day.
Indeed, we partition the~$2^{8192}$ Boolean functions according to their \textsc{robdd}s size (which ranges from~$0$ to~$1277$).

In Figure~\ref{fig:results} the exact distribution is depicted in two ways of presentation:
a red point~$(x, y)$ states that~$2^y$ functions have a \textsc{robdd} size~$x$,
in logarithmic scale;
the blue curve is the probability distribution.

\begin{figure}[htbp]
	\centering{\includegraphics[width=0.65\linewidth]{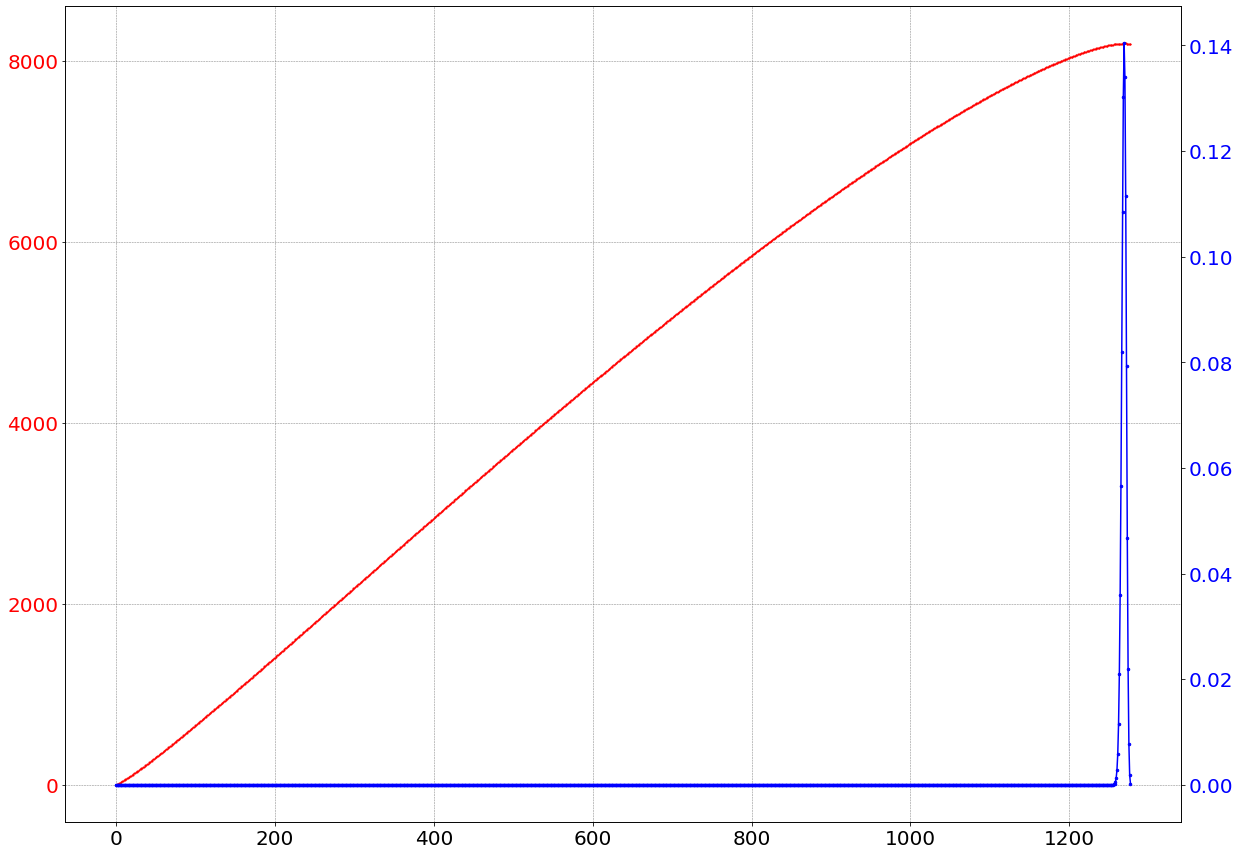}}
	\caption{\label{fig:results}
		The \textsc{robdd} size distribution of Boolean functions in~$13$ variables. The exact distribution is depicted in two ways of presentation: the red curve is the logarithmic scale of the distribution; the blue curve is the probability distribution.}
\end{figure}

\subparagraph*{Organization of the paper.}
In Section~\ref{sec:preliminaries} we present the formal notions and objects necessary for the description of our counting approach.
Section~\ref{sec:counting} presents the iterative process for computing the number of \textsc{robdd}s having a given profile (the profile describes the number of decision nodes labelled with each variable).
Section~\ref{sec:counting} also states the main result of this paper which is a formula for computing the number of \textsc{robdd}s involving linear maps over a polynomial ring.
Finally, Section~\ref{sec:algos} presents the algorithmic context for computing the complete distribution (under the form of the enumerative generating function \cite{FlajoletSedgewick2009} of \textsc{robdd}s).

\section{Preliminaries}
\label{sec:preliminaries}
A Boolean function in~$k$ variables is a function from the set~$\{\False, \True\}^k$ into the set~$\{\False, \True\}$.
The set of functions is denoted by~$\B_k$ and
its cardinality is~$2^{2^k}$.
Furthermore, for the rest of this paper, we choose an ordering of the sequence of variables that corresponds to~$x_1, x_2, \dots, x_k$. 
Any other ordering could be chosen, but one must be fixed.

\subsection{Boolean functions representation}
\label{sec:reprentatives}

Figure~\ref{fig:example} shows a decision tree representing a 4-variable Boolean function and its associated \textsc{robdd}. In both structures, traversing a path from the root to a leaf allows to evaluate the function for a given assignment.
Being in a node labelled by~$x_i$ and going to its low child (using the dotted edge) corresponds with evaluating~$x_i$ to \False; going to its high child (using the solid edge) corresponds with evaluating~$x_i$ to \True.

\begin{figure}[htb]
	\centering{\includegraphics[height=3.5cm]{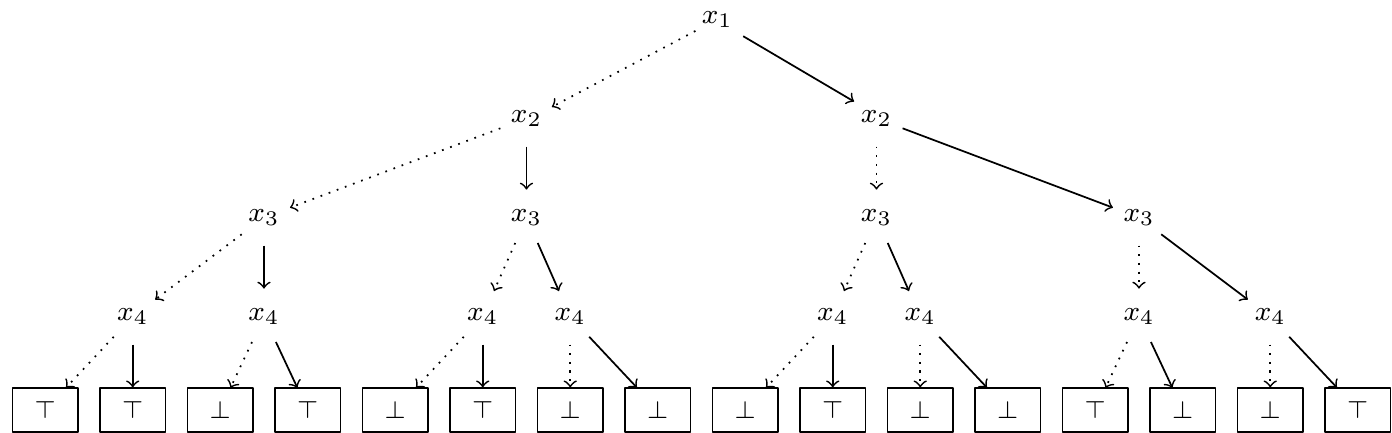}\hfill
		\includegraphics[height=3.5cm]{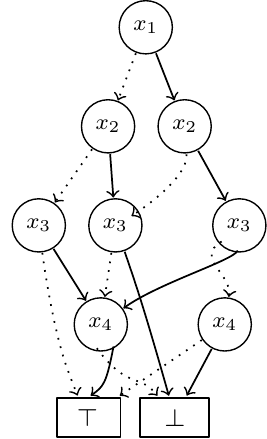}}
	\caption{\label{fig:example} \textbf{(left)} A decision tree and \textbf{(right)} its associated \textsc{robdd}.}
\end{figure}

\begin{definition}[\textsc{dag} representation]
	Let~$f\in \B_k$ be a Boolean function in variables~$x_1, \dots, x_k$.	
	The function~$f$ can be represented as a rooted \emph{directed acyclic graph} (\textsc{dag} for short) composed of internal \emph{decision nodes}, labelled by variables
	and two \emph{terminal nodes} labelled by~$ \{\Fnode, \Tnode\}$ representing respectively the constants~$\False$ and~$\True$.
	Each decision node labelled by~$x_i$ has two children, the low child (resp. high child) such that traversing the edge to
	the low child (resp. high child) corresponds to assign~$x_i$ to~$\False$ (resp. to~$\True$).
	The \emph{size} of a \textsc{dag} is its number of decision nodes. 
  \lipicsEnd
\end{definition}

\begin{definition}[\textsc{obdd}]
	Let~$f\in \B_k$ and pick one of its \textsc{dag} representation.
	The \textsc{dag} is called an Ordered Binary Decision Diagram for~$f$ (\textsc{obdd} for short) when all paths from the root to a terminal node traverse decision nodes with index in increasing order.
  \lipicsEnd
\end{definition}

By taking an \textsc{obdd} for a function with a pointed decision node~$\nu$ labelled by~$x_i$, and extracting the sub-\textsc{dag} rooted in~$\nu$ by taking all its descendants, we obtain an \textsc{obdd} representing a Boolean function in the variables~$x_i, x_{i+1}, \dots, x_k$.

\begin{definition}[\textsc{robdd}]
	Let~$f\in \B_k$ and let $B$ be one of its \textsc{obdd}.
	If all sub-\textsc{dag}s of~$B$ are representing distinct Boolean functions, then~$B$
	is called Reduced Ordered Binary Decision Diagram (\textsc{robdd}).
  \lipicsEnd
\end{definition}

Figure~\ref{fig:merge-delete} shows the forbidden configurations in \textsc{robdd} and the operations used to compress the structure. 

\begin{figure}[htb]
\centering{\includegraphics[height=3cm]{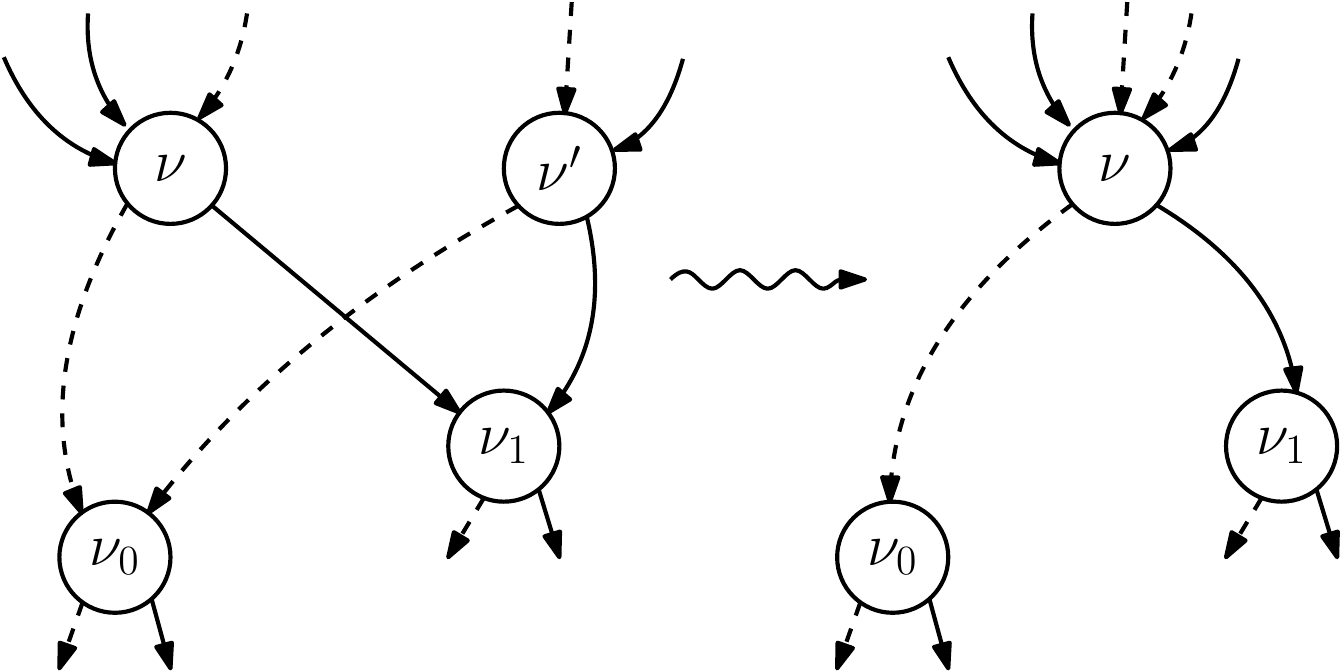}\hspace{2.5cm}
  \includegraphics[height=3cm]{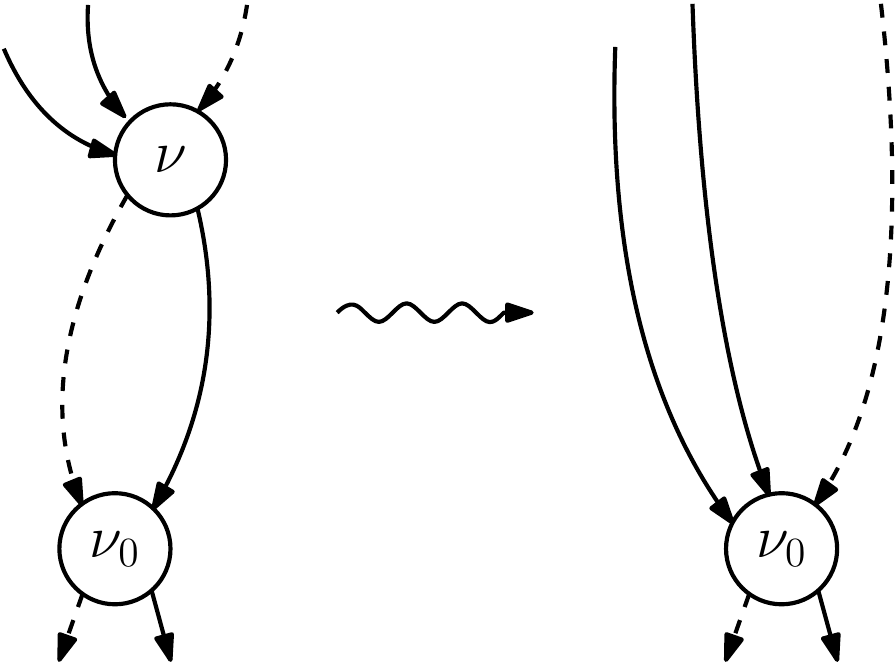}}
  \caption{\label{fig:merge-delete} The two forbidden configurations in \textsc{robdd}s and the resulting operations when compressing: \textbf{(left)}~$\nu$ and~$\nu'$ are merged; \textbf{(right)}~$\nu$ is deleted (from \cite{GPS04}).}
\end{figure}

\begin{fact*}
	Let~$f\in \B_k$, there exists a single \textsc{robdd} representing~$f$.
	\lipicsEnd
\end{fact*}
The data structure of \textsc{robdd}s is especially famous due to this property of 
canonicity\footnote{Uniqueness of the \textsc{robdd} for~$f$ is ensured when a variable ordering is fixed.}.
The reader may read Knuth~\cite{Knuth11} for the proof of uniqueness and several other properties
satisfied by \textsc{robdd}s.
	
For the rest of the paper we are only dealing with \textsc{robdd}s.
Furthermore, we take the point of view of layered \textsc{robdd} by studying and representing
\textsc{robdd}s layer by layer: each layer contains all decision 
nodes labelled by the same variable. This layer-by-layer decomposition is natural since the variables
are ordered. Thus, a \textsc{robdd} for a function in~$k$ variables
is composed of~$k$ layers (plus a layer for the two terminal nodes). Note that a layer could be
empty (which means the Boolean function does not depend on the particular variable associated to this layer).

\begin{definition}[\textsc{robdd}'s profile]
Let~$f\in \B_k$ and let $B$ be its \textsc{robdd}.
Using the layer-by-layer point of view, we define the profile~$B$ to be the non-negative integer sequence~$[p_1, p_2, \dots, p_k]$ such that~$p_i$ is the number of nodes labelled by~$x_i$ in~$B$, i.e. the size of the layer corresponding to~$x_i$, for all~$i\in\{1, \dots, k\}$.
\lipicsEnd
\end{definition}

\subsection{Combinatorial description}

In our previous paper~\cite{CG20} we proposed a kind of recursive decomposition of a \textsc{robdd} based on the low and high children of the root. 
%Intuitively taking a decision node inside a \textsc{robdd} we can define the \textsc{robdd} hanging to this decision node.
%However, since the high child can point within a node of the sub-\textsc{dag}s of the low child, we had to take into account this dependency.
Here we propose a new and simpler point of view, based on a layer-by-layer description, which is much more efficient
for the counting problem (also for the generating problem). 
%Using our combinatorial description of a \textsc{robdd} in \emph{layers} we obtain 
%a convenient formal characterization of the decomposition
%of \textsc{robdd}s level by level. 
We need to introduce a generalization of a \textsc{robdd}, called multientry \textsc{robdd}s, which corresponds exactly to the structure obtained by removing some upper layers in a standard \textsc{robdd} (see Figure~\ref{fig:multi} for an example).

Informally a multientry \textsc{robdd} is a structure obtained by cutting off a certain number of the top layers of a \textsc{robdd}. The resulting structure is still a \textsc{dag}, but with several sources. We also keep track of where the half-edges from the top layers were pointing. In Figure~\ref{fig:multi} a \textsc{robdd} and multientry \textsc{robdd} obtained by removing the~$3$ top layers are depicted.

\begin{definition}[Multientry \BDD]
\label{def:multi}
	A multientry \BDD\ $M$ with at most~$k\ge 0$ variables and size~$n\ge 0$ is a couple~$(M, \mathfrak{E})$. 
	The structure~$M$ is a layered \textsc{dag} structure with~$n$ decision nodes~$\mathcal{Q}$ distributed on~$k$ layers
	and two constant nodes~$\{\top, \bot\}$ such that any two subgraphs are non-identical (as in \textsc{robdd}s). 
	The multiset~$\mathfrak{E}$ has its elements in~$\mathcal{Q}\cup \{\top, \bot\}$ and is such that 
	any source node in~$M$ (i.e., having in-degree~$0$), must appear at least once.
	The nodes in~$\mathfrak{E}$ are called \emph{distinguished} (and correspond with destination nodes of red half-edges in Figure~\ref{fig:multi}).
\lipicsEnd
\end{definition}
As a special case, a \textsc{robdd} is a multientry \textsc{robdd} having one source (the root) and a multiset~$\mathfrak{E}$ reduced to the root (with multiplicity~$1$).

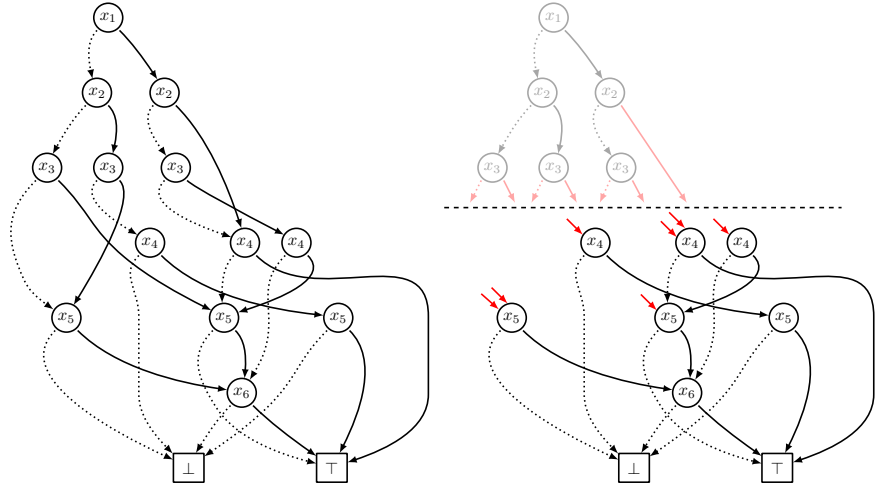
\begin{figure}[htb]
\begingroup
	\scalebox{0.6}{\input{figs/fig-black-6_15_2289425208767.tikz}}
	\scalebox{0.6}{\input{figs/fig-black-spine_6_15_2289425208767.tikz}}
\endgroup
	\caption{\label{fig:multi}\small A \textsc{robdd} of size~$13$ (not counting the constant nodes~$\{\Fnode$ and~$\Tnode\}$),
		with~$6$ variables and profile~$\boldsymbol{p}=\SEQ{1, 2, 3, 3, 3, 1}$. \textbf{(left)} A \textsc{dag} representation of the \textsc{robdd};
		\textbf{(right)} Cutting off top three layers, we obtain the multientry \textsc{robdd} 
	  	keeping track with red half edges of nodes which were disconnected.}
\end{figure}
In the multientry \BDD\ in Figure~\ref{fig:multi}, by numbering the nodes from top to bottom and from left to right,
i.e. the leftmost~$x_4$ is number 1, the second~$x_4$ is 2, the rightmost~$x_4$ is 3, the rightmost~$x_5$ is 4, \dots,
$x_6$ is 7 and~$\bot$ and~$\top$ are respectively 8 and 9, we obtain the multiset 
$\mathfrak{E} = \{1, 2, 2, 3, 4, 4, 5\}$.
We note that our definition of multientry \BDD\ is similar (but not exactly identical) to the one of
\emph{shared-}\textsc{BDD}s presented by Knuth \cite{Knuth11} to represent several Boolean functions in the same decision diagram.

Furthermore, remark that the same multientry \BDD\ can be exhibitted by cutting off top layers (not even the same number) of different \textsc{robdd}s.

\section{Full iterative counting formula}
\label{sec:counting}

In the following we describe an approach to counting the number of \textsc{robdd}s
with a given profile $\boldsymbol{p}$. % size~$n$ when~$k$ is growing, for all possible values for~$n$.
We use a powerful algebraic representation to encapsulate a kind
of inclusion-exclusion principle. This motivates the definition of a linear application on polynomials using substitutions. 
% In this section we adopt a pedagogical way for presenting Theorem~\ref{thm:main} applied to a specific multientry \textsc{robdd} and in the example which follows the statement of the theorem. 
The linearity property of the applications is crucial for achieving our algorithm polynomial complexity.  We consider the polynomial ring~$\mathbb{Z}[X]$ and linear endomorphisms over~$\mathbb{Z}[X]$ (i.e., linear maps between~$\mathbb{Z}[X]$ and~$\mathbb{Z}[X]$). Thus, for a linear map~$g:  \mathbb{Z}[X] \to  \mathbb{Z}[X]$ and  two polynomials~$P$ and~$Q$ in~$\mathbb{Z}[X]$,~$g[P+Q]=g[P]+g[Q]$, and for any scalar~$\lambda \in \mathbb{Z}$ and polynomial~$P \in \mathbb{Z}[X]$
we have~$g[\lambda P] = \lambda g[P]$.
% But at last we will rely on the linearity property of all applications we need so that the global distribution computation (Section~\ref{sec:algos}) will not consider individual multientry \textsc{robdd}s anymore.

In the following theorem, we state the main result of the paper which gives access 
to the number of multientry \textsc{robdd}s for a given profile.
\begin{theorem}[Multientry \textsc{robdd}s counting formula]
\label{thm:main}
	For the family of linear maps~$(\phi_r)_{r\ge 0}$,
	each mapping~$\phi_r : \mathbb{Z}[X] \to \mathbb{Z}[X]$ is defined with respect to the canonical basis~$(X^m)_{m\ge 0}$ by 
	\begin{equation}
	  \label{eq:Rrm}
	  \phi_{r}[X^m] = \left(\prod_{i=0}^{r-1} \left(X^2-X-i\right)\right) \cdot
	  \left(\sum_{j=0}^{m-r} \binom{m}{j} \stirlingII{m-j}{r} \, X^j \right) .
	\end{equation}
	Let~$M(\boldsymbol{p}, m)$ be the number of multientry
	\textsc{robdd}s with profile~$\boldsymbol{p}=\SEQ{p_1, \dots, p_k}$ 
	and~$m$ incoming half edges.
	We have, for~$k\ge 0$,
	\begin{equation}
	  \label{eq:subs}
	  M(\boldsymbol{p}, m) = \left(\phi_{\boldsymbol{p}}[X^m]\right)_{X=2},
	\end{equation}
	where
	\begin{itemize}
	\item
	 $\phi_{\boldsymbol{p}}$ is the composition product~$\phi_{p_k} \circ \phi_{p_{k-1}} \circ \dots \circ \phi_{p_{1}}$;
	\item
	  for a polynomial~$P \in \mathbb{Z}[X]$,~$(P)_{X=2}$ is the evaluation of~$P$ at~$X=2$.
	\end{itemize}
\lipicsEnd
\end{theorem}
In the theorem,~$\binom{m}{j}$ stands for the binomial coefficient
and~$\stirlingII{n}{k}$ is the Stirling number of the
second kind counting the number ways to partition a set of~$n$ objects into~$k$
non-empty subsets.

\begin{remark}
\label{rem:basic}
	This is actually a stronger result that what we need for counting \BDDs, 
	since the number of \BDDs\ corresponds to the special case where
	$\boldsymbol{p}=\SEQ{p_1=1, p_{2}, \dots, p_k}$ and 
	$m=1$ (meaning there is one source node in the top layer). 
	However, the fact that we are able to compute 
	$\phi_{\boldsymbol{p}}$ on the basis $(X^m)_{m\ge 0}$ is key to our approach. 
\end{remark}
The detailed the proof of Theorem~\ref{thm:main} is deferred to after an example
of such a computation. %presented in Appendix~\ref{app:counting}.
\begin{example*}
Let us consider a profile~$\boldsymbol{p}=\SEQ{1, 2, 4, 2}$ for~$4$ variables~$x_1, x_2, x_3, x_4$ and choosing~$m=1$ in \eqref{eq:subs}. Then, for~$1\le i \le 4$, we compute iteratively
$\phi_{p_i} \circ \dots \circ \phi_{p_1}(X)$:
\begin{footnotesize}
  \begin{align*}%\footnotesize
  X & \xmapsto{\phi_1}  X^2-X\\
	&\quad \xmapsto{\phi_2} X^4 - 2X^3 + X\\
	&\qquad \xmapsto{\phi_4} X^8 - 4X^7 + 14X^5 - 6X^4 - 16X^3 + 5X^2 + 6X\\
	&\qquad  \quad \xmapsto{\phi_2} 28X^{10} + 28X^9 - 98X^8 - 112X^7 + 76X^6 + 92X^5 - 12X^4 - 8X^3 + 6X^2.
\end{align*}
\end{footnotesize}
Evaluating the last polynomial at~$X=2$, we get that there are~$11\,160$
\textsc{robdd}s with profile~$\SEQ{1, 2, 4, 2}$.
The power of our approach is that we could have stopped at any iteration, and still get the number of \textsc{robdd}s for the considered truncated profile. On the particular example this yields
\par
\noindent\begin{tabularx}{\linewidth}{l|l|l}
  \footnotesize
  $\boldsymbol{p}$ & $\phi_{\boldsymbol{p}}$ &  $\left(\phi_{\boldsymbol{p}}\right)_{X=2}$ \\
  \hline
  $\SEQ{\,}$ & $X$ & $2$\\
$\SEQ{1}$&   $X^2-X$ & $2$\\
$\SEQ{1, 2}$&   $X^4 - 2X^3 + X$ & $2$ \\
$\SEQ{1, 2, 4}$&   $X^8 - 4X^7 + 14X^5 - 6X^4 - 16X^3 + 5X^2 + 6X$ & $0$\\
$\SEQ{1, 2, 4, 2}$ & $28X^{10} + 28X^9 - 98X^8 - 112X^7 + 76X^6 + 92X^5 - 12X^4 - 8X^3 + 6X^2$ & $11\,160$\\
\end{tabularx}
\par\noindent
The number~$0$ when considering~$\SEQ{1, 2, 4}$ may seem counterintuitive at first, but indeed a~$\textsc{robdd}$ can only have up to~$2$ nodes on its last layer, otherwise one node has to be a duplicate of another. 
\end{example*}
\begin{proof}[Proof of Theorem~\ref{thm:main}](Multientry \textsc{robdd}s counting formula).
% Julien: Je commente. Je sais plus où ça sert.
% We will exploit the following properties related to the family $(\phi_{r}[X^m])_{r, m \ge 0}$:
% \begin{itemize}
% \item
%   $\phi_0$ is the identity (since $P_{0,m}(X)=X^m$ for $m\ge 0$), implying by linearity that $\phi_0[P]=P$ for any polynomial $P$.
% \item  when $r>0$, for any polynomial $P$ and any constant $c\in \mathbb{Z}$ we have $\phi_r[P+c]=\phi_r[P]$\marginpar{\small J: Je sais plus où on s'en sert...}. Indeed, if $r>0$ and for any $m\ge 0$
%   \[
%     \left(\phi_{r}[X^m]\right)_{X=0} = \left(\phi_{r}[X^m]\right)_{X=1}=0,
%   \]
%   since from \eqref{eq:Rrm} in this case $X^2-X=X(X-1)$ divides $\phi_{r}[X^m]$. So for any constant $c\in \mathbb{Z}$, $\phi_r[c] = 0$ whenever $r>0$ and by linearity, we get the result.
% \end{itemize}
The proof is obtained by induction on the number~$k\ge 0$ of layers with decision nodes.

\proofsubparagraph{Base case.}
When~$k=0$ and~$n\ge 0$. The number of multientry \textsc{robdd}s is~$M(\SEQ{\,}, m) = 2^m$, i.e.,~$X^m$ evaluated at~$X=2$, as we must map the~$m$ half edges to either one of the two constants. Note that if~$m=0$ then~$M(\SEQ{\,}, m) = 1$ corresponding to the void function (which is a special case).

\proofsubparagraph{Induction step.}
Now suppose Theorem~\ref{thm:main} is true for~$k\ge 0$.

Let us consider a profile of length~$k+1$ as~$\SEQ{r} \cdot \boldsymbol{p}$ with 
$r\ge 0$ and~$\boldsymbol{p}$ a profile of length~$k$.
\begin{itemize}
\item
  If~$r=0$, a simple computation shows that~$\phi_{0}$ is the identity. Hence, the empty layer can in fact be omitted since
\begin{equation}
  \label{eq:recursiveM0}
  \phi_{\boldsymbol{p}} \circ \phi_{0}\left[X^{m}\right]= \phi_{\boldsymbol{p}} \left[X^{m}\right],\quad\text{so that $M(\SEQ{0} \cdot \boldsymbol{p}, m)=M(\boldsymbol{p}, m)$}.
\end{equation}
\item
From now on let us suppose~$0 <r \le m$.
The set of~$m$ half edges pointing at the first layer can be decomposed 
in two subsets for~$j\in\{0, \dots, m-r\}$:~$j$ entries will go to layers
below the first one, and~$m-j$ entries will be mapped to
the~$r$ nodes of the first layer. A Stirling number of the
second kind~$\stirlingII{n}{k}$ counts the number of ways to partition 
a set of~$n$ objects into~$k$ non-empty subsets. 
So there are~$\binom{m}{j} \cdot \stirlingII{m-j}{r}$ such partitions. 
We write
\begin{equation}
  \label{eq:decompositionM}
  M(\SEQ{r}\cdot\boldsymbol{p}, m) =
  \sum_{j=0}^{m-r} \binom{m}{j} \stirlingII{m-j}{r}
  f^{(r)}_{\boldsymbol{p}}(j),
\end{equation}
where~$f^{(r)}_{\boldsymbol{p}}(j)$ denotes the number of multientry \textsc{robdd}s with~$j$ free half edges and~$r$ pairs of half edges (the ones resulting from the~$r$ nodes of the first layer). These~$2r$ half edges must thus obey the following constraints:
\begin{itemize}
\item
  in each pair, the two half edges must be distinct, i.e., point to different nodes;
\item
  all~$r$ pairs of half edges must be distinct with one another (as pairs).
\end{itemize}
Our goal now is to get rid of the constraints coming from these~$r$ nodes and express all quantities in terms of free half edges.

The following equation translates the previous constraints on the pair of adjacent half edges coming from the first of the~$r$ nodes
\begin{align}
  \label{eq:recursiveM}
  f^{(r)}_{\boldsymbol{p}}(j) 
  = f^{(r-1)}_{\boldsymbol{p}}(j+2)  
   -f^{(r-1)}_{\boldsymbol{p}}(j+1)  
    - (r-1)  f^{(r-1)}_{\boldsymbol{p}}(j).
\end{align}
Indeed, the first term~$f^{(r-1)}_{\boldsymbol{p}}(j+2)$ corresponds in adding~$2$ free half edges, and results in overcounting. Then 
following an inclusion-exclusion principle, firstly we subtract 
$f^{(r-1)}_{\boldsymbol{p}}(j+1)$ 
which would count the number of configurations if the two half edges in this pair were merged. 
Finally, we subtract~$(r-1)  f^{(r-1)}_{\boldsymbol{p}}(j)$
in~\eqref{eq:recursiveM}. This last quantity counts the number of configurations if the pair of half edges was merged with one of the~$r-1$ remaining ones (hence~$r-1$ choices). Note that no additional free half edge is added because of this merge. This yields~\eqref{eq:recursiveM}.

Solving the simple recurrence~\eqref{eq:recursiveM} with respect to~$r$ yields
\begin{equation}
  \label{eq:induction}
   f^{(r)}_{\boldsymbol{p}}(j) =  \sum_{i=0}^{2r}  a_i f^{(0)}_{\boldsymbol{p}}(i+j),
\end{equation}
where coefficients~$(a_i)$ are obtained by identifying~$P(X) = \prod_{i=0}^{r-1}(X^2-X-i) = \sum_{i=0}^{2r} a_i X^i$.

At this point, we remark the equality true for~$m\ge 0$
\begin{equation}
  \label{eq:free}
  M(\boldsymbol{p}, m) = f^{(0)}_{\boldsymbol{p}}(m),
\end{equation}
which reflects the fact that in~$M(\boldsymbol{p}, m)$, all~$m$ half edges are unconstrained.
Then \eqref{eq:decompositionM} rewrites 
\begin{align*}
  M(\SEQ{r}\cdot\boldsymbol{p}, m)
  = \sum_{j=0}^{m-r} \binom{m}{j}  \stirlingII{m-j}{r} \sum_{i=0}^{2r} a_i M(\boldsymbol{p}, i+j)
   = \sum_{i=0}^{m+r} c_i M(\boldsymbol{p}, i),
\end{align*}
with coefficients~$c_i$ obtained by identifying~$\sum_{i=0}^{m+r} c_i X^i = \phi_r(X^m)$ from \eqref{eq:Rrm}.

By induction hypothesis on the length~$k$ of~$\boldsymbol{p}$ we have for~$0\le i\le m+r$
\begin{equation}
  \label{eq:induction_hypothesis}
  M(\boldsymbol{p}, i) =  \left(\phi_{\boldsymbol{p}}\left[X^{i}\right]\right)_{X=2}.
\end{equation}
By linearity
\begin{align*}
\sum_{i=0}^{m+r} c_i  \phi_{\boldsymbol{p}} \left[X^{i}\right] 
  = \phi_{\boldsymbol{p}} \left[\sum_{i=0}^{m+r} c_i X^i\right]
  = \phi_{\boldsymbol{p}} \left[\phi_r[X^m]\right]=\phi_{\boldsymbol{p}} \circ \phi_r[X^m],
\end{align*}
and finally
\begin{align*}
  M(\SEQ{r}\cdot\boldsymbol{p}, m)
   = \sum_{i=0}^{m+r} c_i M(\boldsymbol{p}, i)
    = \sum_{i=0}^{m+r} c_i  \left(\phi_{\boldsymbol{p}} \left[X^{i}\right]\right)_{X=2}
  = \left(\phi_{\boldsymbol{p}} \circ \phi_r[X^m]\right)_{X=2}
\end{align*}
This ends the proof.
\end{itemize}
\end{proof}

\section{Counting algorithms}
\label{sec:algos}

In this section, we present an algorithm for counting \textsc{robdd}s of size~$n$. 
%However, recall that for~$k$ variables, the number of Boolean functions is~$2^{2^k}$.

\emph{The time and space complexities are measured respectively in terms of arithmetical operations on~$\mathbb{Z}$ and memory space used to store integers in~$\mathbb{Z}$.  When considering \textsc{robdd}s of size upper bounded by~$n$, all integers in~$\mathbb{Z}$ involved can be checked to be of bit length~$O(n\log n)=O(\log n!)$.}
%\noteJul{Peut-être qu'on pourrait dire que les entiers en général sont de taille~$O(n \log n)$, parce que c'est pénible à rédiger.}
%\MP{mettre github à jour}

The reader can find an implementation of the following algorithms at \url{https://github.com/agenitrini/BDDgen}.

\subsection{Linear maps: precomputation step}
A first  step is to pre-compute a representation of linear maps~$(\phi_r)_{r\ge 0}$. For a \textsc{robdd} of size~$n$ we know that the maximal number of half edges is~$n+1$, and the maximal number of nodes on a layer is also loosely upper bound by~$n$. Hence, it is sufficient to compute~$\phi_r[X^m]$ in~$\mathbb{Z}[X]$ for~$0\le r \le n$ and~$0 \le m \le n+1$. In the form of \eqref{eq:Rrm},~$\phi_{r}[X^m]$ is equal to~$P_r(X) Q_{r, m}(X)$ with
\[
  P_r(X) =
  \prod_{i=0}^{r-1} (X^2-X-i), \quad\text{ and } \quad Q_{r, m}(X)=\sum_{j=0}^{m-r} \binom{m}{j}\stirlingII{m-j}{r}X^j.
\]
So the first step is to compute coefficients of~$P_r(X)$ and~$Q_{r, m}(X)$.
Concerning binomial coefficient and Stirling numbers of the second kind, both tables can be computed by a naive algorithm (for binomials, it is the famous Pascal's triangle) in space~$O(n^2)$ with~$O(n^2)$ arithmetic operations on integers.

Once these coefficients are available, we compute the products~$\phi_{r}[X^m] = P_r(X) Q_{r, m}(X)$.
Computing~$P_r$ from~$P_{r-1}$ necessitates~$O(n)$ arithmetical operations on integers, yielding a total~$O(n^2)$ number of arithmetical operations for the whole family~$(P_r)_{r\le n}$. Each polynomial~$(Q_{r, m})$ necessitates~$O(n)$ arithmetical operations per polynomial (supposing binomial coefficients and Stirling number of the second kind are precomputed). Finally,~$\phi_r[X^m]$ is computed with~$O(n^2)$ arithmetical operations (by a naive product of two polynomials). There are~$O(n^2)$ such polynomials to compute.  Hence, we get a total~$O(n^4)$ of arithmetic operations using~$O(n^3)$ storage memory space for coefficients. We thus get the next lemma.
\begin{lemma}[Precomputing step: linear maps]
  \label{lem:linear}
  The precomputation step for the representation of linear maps by computing~$\phi_r[X^m]$ for~$0 \le r \le n$ and~$0\le m\le n+1$ necessitates~$O(n^4)$ arithmetical operations on integers and uses memory space for~$O(n^3)$ coefficients in~$\mathbb{Z}$.
  \lipicsEnd
\end{lemma}
  
\subsection{Basic counting}
The basic block in our approach is to be able to compute~$\phi_r[P]$ for
$r\ge 0$. This is done by Algorithm~\ref{algo:iterate} which is a direct translation of Theorem~\ref{thm:main} using the linearity of the linear maps~$(\phi_r)_r$.
%\MP{modifier downto dans algo1}
\begin{figure}[htb]
	\begin{minipage}[t]{0.45\linewidth}
		\begin{algorithm}[H]
			\footnotesize
			\caption{\footnotesize Computing~$\phi_r[P]$}
			\label{algo:iterate}
			\KwIn{an integer~$r\ge 0$, a polynomial~$P(X)=\sum_{m=0}^{d} p_m X^m \in \mathbb{Z}[X]$ }
			\KwOut{$\phi_{r}[P] \in \mathbb{Z}[X]$}
			$Q \gets 0$\;
			\For{$m$ \From $r$ \KwTo $d$}{
				$Q \gets Q + p_m \phi_r[X^m]$\;
			}
			\Return{$Q$}
		\end{algorithm}
	\end{minipage}\hfill
	\begin{minipage}[t]{0.52\linewidth}
		\begin{algorithm}[H]
                  \footnotesize
                  \caption{Computing $(\phi_{\boldsymbol{p}}[X])_{X=2}$}
                  \KwIn{a profile~$\boldsymbol{p}=\SEQ{p_1, \dots, p_k}$}
                  \KwOut{the number of \textsc{robdd}s with profile~$\boldsymbol{p}$}
			\label{algo:count}
			$P \gets X$\;
			\For{$i$ \From $1$ \To $k$}{
                          $P \gets \phi_{p_i}[P]$
			}
			\Return{$P(2)$}
			%\vspace*{21.8pt}
		\end{algorithm}
	\end{minipage}
\end{figure}

%\par\noindent
%\begin{algorithm}
%\small
%  \caption{iterate($r, P$)}
%  \label{algo:iterate}
%  \KwIn{ integer $r\ge 0$ and a polynomial $P(X)=\sum_{i=0}^{d} a_i X^i \in \mathbb{N}_0[X]$ }
%  \KwOut{Polynomial $\phi_{v}(P) \in \mathbb{Z}[X]$}
%  $Q \gets 0$\;
%  \For{$m$ \From $r$ \KwTo $d$}{
%    $Q \gets Q + a_m R_{r,m}$\;
%  }
%  \Return{$P$}
%\end{algorithm}
\begin{proposition}[Complexity of basic step]
\label{prop:basicStep}
	Let $P$ be a polynomial of degree~$d$. Algorithm~\ref{algo:iterate} computes~$\phi_r[P]$ and performs
	$O(rd +d^2)$ arithmetical operations over~$\mathbb{Z}$ to compute~$\phi_r[P]$, 
	using~$O(d^2)$ memory space\footnote{Recall in Proposition~\ref{prop:basicStep}
		 we do not take into account the precomputation step of the family~$(\phi_r[X^m])_{r}$.}.
	 \lipicsEnd 
\end{proposition}

\begin{proof}
  	Each polynomial~$\phi_r[X^m]$ is of degree~$r+m=O(r+d)$. 
	Thus, the number of operations needed on coefficients is~$O(rd+d^2)$ if~$r>0$
	(or~$O(1)$ if~$r=0$ since~$\phi_0$ is the identity).
      \end{proof}

By Remark~\ref{rem:basic} and applying Theorem~\ref{thm:main},
it is straightforward to compute the number of \textsc{robdd}s
with profile~$\boldsymbol{p}=\SEQ{p_1, \dots, p_k}$ 
(using~$m=1$).
The pseudocode is given in Algorithm~\ref{algo:count}. We have the following proposition.
\begin{proposition}
\label{prop:MComplexity}
	Algorithm~\ref{algo:count} computes the number of \textsc{robdd}s of size~$n$ with~$k$ variables
        and given profile~$\boldsymbol{p}$. It performs~$O(k\;n^2)$ arithmetical operations over~$\mathbb{Z}$ and uses~$O(n)$ extra memory space to store integers.
        \lipicsEnd
\end{proposition}
\begin{proof}
  The number of \textsc{robdd}s is~$(\phi_{\boldsymbol{p}}[X])_{X=2}$. 
  The polynomial~$\phi_{\boldsymbol{p}}[X]$ is computed by iterating~$k$ times a linear map of type~$\phi_r$ starting from an initial polynomial~$X$.    
  By Proposition~\ref{prop:basicStep}, the algorithm performs~$O(n^2)$ operations 
  for each iteration since polynomials have~$O(n)$ coefficients. % (as can be easily checked).
  The total computation thus performs~$O(k\; n^2)$ arithmetical operations over~$\mathbb{Z}$ and use~$O(n)$ memory space to store coefficients.
  Evaluating~$(\phi_{\boldsymbol{p}}[X^m])_{X=2}$ at~$X=2$ can be done in time complexity~$O(n)$ (by Horner's method for instance).
\end{proof}

\subsection{Generating function for ROBDD size} 
The main goal of this section is compute the distribution of Boolean functions in~$k$ variables
according to the \textsc{robdd} size. %Let us denote $\mathcal{B}_k$ the set of Boolean functions on $k$ variables and 
For~$f\in\mathcal{B}_k$ a Boolean function, we let~$\lambda(f)$ be the size of its \textsc{robdd}, i.e., its number of decision nodes.
A convenient way to represent the distribution of the size~$\lambda$ on~$\mathcal{B}_k$ consists in computing
the generating function \cite{FlajoletSedgewick2009}
\[
  F_k(u)  = \sum_{f \in \mathcal{B}_k} u^{\lambda(f)} = \sum_{i\ge 0} f_i u^i,
\]
where~$u$ is a formal variable marking the size. Then for $i \ge 0$, the coefficient $f_i=[u^i] F_k(u)$ is the number of \textsc{robdd}s of size~$i$ with~$k$ variables, i.e. the notation $[u^i] F_k(u)$ corresponds to the coefficient-extraction of the monomial $u^i$. We also introduce the truncation~$F_k^{\le n}(u)= \sum_{0 \le i\le n} f_i u^i$ as the generating functions of \textsc{robdd}s of size less than or equal to~$n$.

We first extend the formalism introduced in Section~\ref{sec:counting} and define a linear map~$\varphi:\mathbb{Z}[u,X] \to \mathbb{Z}[u, X]$.
\begin{definition}
  \label{def:varphi}
The linear map~$\varphi:\mathbb{Z}[u,X] \to \mathbb{Z}[u, X]$ is defined via its action on the basis~$(u^rX^m)_{r, m \ge 0}$ as
\begin{equation}
  \label{eq:varphi}
 \varphi: u^rX^m \mapsto \varphi[u^rX^m] = \sum_{i=0}^{m+1} u^{r+i} \phi_i[X^m]. %=Q_{r, m}(u, X).
\end{equation}
\lipicsEnd
\end{definition}
With this notation we have the following proposition.
\begin{proposition}[Generating function for \textsc{robdd} size]
\label{thm:gf}
The generating function enumerating Boolean functions by considering the \textsc{robdd} size is given by
\[
  F_k(u) = \left(\varphi^k[X]\right)_{X=2},
  \quad\text{where~$\varphi^k$ denotes the composition product~$\underbrace{\varphi\circ \dots \circ \varphi}_{\text{$k$ times}}$}.
\]
\lipicsEnd
\end{proposition}
\begin{proof}
Each application of~$\varphi$ corresponds with adding a layer. The formal variable~$u$ marks the number of decision nodes added on the current layer.
\end{proof}
We remark that in practice we can truncate polynomials, keeping only the terms useful along the computation. The key point here is that if we consider \textsc{robdd}s with size bounded by~$n$, we should make all computation modulo~$u^{n+1}$. Indeed, the formal variable~$u$ marks the number of nodes (which is bounded by~$n$). Sections \ref{appendix:iterate} and \ref{appendix:truncate} int the appendix illustrate this point. 

Algorithm \ref{algo:onelayer} computes the bivariate polynomial~$\varphi[X^m]$ for~$m\ge 0$. Algorithm \ref{algo:memo}, computes recursively the iterated (univariate) version~$\left(\varphi^{\ell}[X^m]\right)_{X=2}$ (that is the evaluation at $X=2$). 
\begin{figure}[htb]
  \begin{minipage}[t]{.44\linewidth}
\vspace*{3pt}
 \begin{algorithm}[H]
\footnotesize
 \SetAlgoProcName{Function}{function}
\caption{\footnotesize Computing~$\varphi[X^m]$}
\label{algo:onelayer}
\SetKwFunction{Parcours}{parcoursProfondeur}
\KwIn{An integer~$m \ge 0$}
\KwOut{Returns~$\varphi[X^m] \in \mathbb{Z}[u, X]$}
$Q\gets 0$ \Comment*{$Q\in \mathbb{Z}[u]$}
\For{$r$ \From $0$ \KwTo $m$}{
 $Q \gets Q + u^r \phi_r[X^m]$
}
\Return{$Q$}
\end{algorithm}
\end{minipage}\hfill
\begin{minipage}[t]{.54\linewidth}
 \begin{algorithm}[H]
\footnotesize
 \SetAlgoProcName{Function}{function}
\caption{\footnotesize Computing $\left(\varphi^{\ell}[X^m]\right)_{X=2}$}%\texttt{rec}($\ell, m$)}
\label{algo:memo}
\SetKwFunction{Parcours}{parcoursProfondeur}
\KwIn{Two integers~$\ell, m$}
\KwOut{Returns~$\left(\varphi^{\ell}[X^m]\right)_{X=2} \in \mathbb{Z}[u]$}
\Comment{N.B.: Computations done modulo~$u^{n+1}$ where~$n$ is the maximal size for \textsc{robdd}s}
\lIf*{$\ell=0$}{\Return~$2^m$ \Comment*{base case}}
$Q\gets 0$ \Comment*{$Q\in \mathbb{Z}[u]$}
$R\gets \varphi(X^m)$ \Comment*{Call Alg. \ref{algo:onelayer}}% to get $R\in \mathbb{Z}[X, u]$}
\For{$j$ \From $0$ \KwTo $\degree_X(R)$}{
 $M \gets \left(\varphi^{\ell-1}[X^j]\right)_{X=2}$ \Comment*{Call Alg. \ref{algo:memo}}% to get $M\in \mathbb{Z}[u]$}
 $N \gets [X^j]\  R(u, X)$ \Comment*{$N \in \mathbb{Z}[u]$}
 $Q \gets Q + M \cdot N$
}
\Return{$Q$}
\end{algorithm}
\end{minipage}
\end{figure}
\begin{lemma}
\label{lem:onelayer}
Algorithm \ref{algo:onelayer} computes~$\varphi[X^m] \in\mathbb{Z}[X,u]$, which has~$O(n^2)$ integer coefficients. It performs~$O(n^2)$ arithmetical operations over~$\mathbb{Z}$. % (if we suppose the family~$(\phi_r[X^m])_{r, m}$ is available from a precomputation step).

Algorithm \ref{algo:memo} computes~$\left(\varphi^{\ell}[X^m]\right)_{X=2} \in \mathbb{Z}[u]$, which has~$O(n)$ coefficients. If we omit recursive calls, it performs~$O(n^3)$ arithmetical operations over~$\mathbb{Z}$, using memoization techniques with~$O((n+k)n^2)$ extra memory storage. 
%Overall, since we compute $O(nk)$ polynomials the total number of arithmetical operations on integers in Algorithm  of % (if we suppose $(\varphi[X^\ell])_{ m}$ is available from a precomputation step).
\lipicsEnd
\end{lemma}
\begin{proof}
\label{lem:algos}
  For Algorithm~\ref{algo:onelayer}, we perform~$m=O(n)$ additions of polynomials in~$\mathbb{Z}[X,u]$ with~$O(n)$ terms, yielding~$O(n^2)$ arithmetical operations over~$\mathbb{Z}$. The result is a bivariate polynomial of bounded degree ($n$ for the variable $u$, $2n$ for the variable~$X$) yielding~$O(n^2)$ coefficients. We suppose Algorithm~\ref{algo:onelayer} has access to polynomials~$\phi_r[X^m] \in \mathbb{Z}[X]$ from a precomputation step.

  For Algorithm~\ref{algo:memo}, two ingredients are essential. Firstly we have to truncate polynomials modulo~$u^{n+1}$ so that operations on univariate polynomials have complexity~$O(n)$ for addition and~$O(n^2)$ for multiplication (using the naive multiplication on polynomials).
  Secondly we also use memoization techniques (meaning we compute in lazy manner intermediate results only once and keep it for further reference, at the expense of memory storage). That means that we consider that at the time we compute~$\left(\varphi^{\ell}[X^m]\right)_{X=2}$, the polynomials~$\left(\varphi^{\ell-1}[X^j]\right)_{X=2}$ are available (i.e., their complexity is taken into account independently). By an amortizing argument, the complexity of computing  the complete family of polynomials~$\left(\varphi^{\ell}[X^m]\right)_{X=2}$ ($0 \le \ell \le k$ and~$0\le m \le n+1$) is still~$O(n^3)$ arithmetical coefficients  per polynomial. We need to store~$O(kn^2)$ integer coefficients for memoization of all intermediate polynomials. 
  We also suppose Algorithm~\ref{algo:memo} has access to bivariate polynomials~$(\varphi[X^m] \in \mathbb{Z}[u, X])_{0 \le m \le n+1}$ from a precomputation step which requires~$O(n^3)$ memory space for integer coefficients.
  A subtle point is to understand that we can truncate polynomials at each step in Algorithm~\ref{algo:memo} and still get the correct result: this can be proved by recurrence on~$\ell$ (see also Section~\ref{appendix:truncate} of the appendix for an example). 
\end{proof}

Algorithm~\ref{algo:memo} also computes the generating function of \textsc{robdd}s for size up to $n$ since posing~$\ell=k$ and~$m=1$ we have
\[
  F_{k}^{\le n}(u) = \left(\varphi^{\ell}[X]\right)_{X=2} \mod u^{n+1}.
\]
\begin{theorem}[Algorithm for computing the exact distribution]
  \label{thm:complexityCountAll}
  We compute the generating function~$F^{\le n}_k(u)$ of Boolean functions in~$\mathcal{B}_k$ for \textsc{robdd}s of size less than or equal to~$n$ using~$O(k\; n^4)$ arithmetical operations in~$\mathbb{Z}$ and~$O((k+n)\; n^2)$ for memory space storing integers.
  \lipicsEnd
  \end{theorem}
  \begin{proof}
From Lemmas \ref{lem:onelayer} and \ref{lem:linear}, the overall complexity is dominated by the computation of the family~$(\phi_r[X^m])_{r, m}$ and the calls to Algorithm \ref{algo:memo}.
By Lemma~\ref{lem:onelayer}, each polynomial~$\left(\varphi^{\ell}[X^j]\right)_{X=2}$ is computed with~$O(n^3)$ arithmetic operations on integer coefficients and there are~$O(k\; n)$ such polynomials. Hence, the total number operations over~$\mathbb{Z}$ is~$O(k\; n^4)$. Furthermore, we store~$O(k\; n^2)$ coefficients in~$\mathbb{Z}$ for memoization. We also store~$O(n^3)$ coefficients for the family~$(\phi_r[X^m])_{m, r} \in \mathbb{Z}[X]$. This yields the claimed complexity.
\end{proof}

To evaluate the complete size distribution we need to consider the size of largest \textsc{robdd}s with~$k$ variables.
\begin{theorem}[Maximal size of \textsc{robdd}s]
\label{theo:largest_robdd}
  Let~$k \ge 1$ be an integer, the maximal number of nodes in a \textsc{robdd} 
  with at most~$k$ variables is
  \[
	  M_k=2^{k-\theta}-3+2^{2^{\theta}}, \quad \text{ with~$\theta = 
	  	\left\lfloor \log_2\left(k - \left\lfloor \log_2\left(k\right)\right\rfloor\right)\right\rfloor$}.
  \]
  The generating function~$F_k(u)$ of Boolean functions in~$\mathcal{B}_k$ according to the \textsc{robdd} size can be computed with~$O(2^{4k}/k^3)$ arithmetical operations in~$\mathbb{Z}$ and uses~$O(2^{2k}/k)$ space.
  \lipicsEnd
\end{theorem}
\begin{proof}
	Note that this formula is equivalent to the one given without proof by 
	Pontus von Brömssen~\cite{oeisA327461}. The existence of~$\theta$ is proved in~\cite{NV19} and from there we can derive the explicit expression of~$\theta$. Then substituting~$n=M_k \approx 2^k/k$ in~Theorem~\ref{thm:complexityCountAll} yields the complexity result.
\end{proof}
Note this is the polynomial with respect to the maximal size of a \textsc{robdd} for~$k$ variables, hence a huge improvement compared to exponential brute force algorithms enumerating all~$2^{2^k}$ Boolean functions and computing their \textsc{robdd} size obtained after a compaction process. With a careful implementation, we can achieve the computation of~$F_k(u)$ for~$k$ up to~$11$ variables on a personal computer, and, on a high-performance computer with~$512$ GB RAM memory, for~$k=12$ in less than 1 hour 30 minutes, and even for~$k=13$ in  less than 30 hours using the \texttt{PyPy} implementation of \texttt{Python}.

\section{Conclusion}
As an application of the counting approach of this paper
we are able sample at random \textsc{robdd}s. More precisely, we can efficiently and uniformly pick \textsc{robdd}s either
according to a given size, or even a given profile or a given spine.
This is a great improvement when comparing to the classical uniform random generation over the set
of Boolean functions, like in~\cite{NV19}, that is drastically biased to the largest \textsc{robdd}s
due to the Shannon effect. For instance with~$12$ variables, the probability of drawing uniformly 
a Boolean function giving a \textsc{robdd} of (quadratic in~$k$) size~$144=12^2$ is approximately~$1.212 \cdot 10^{-957}$.

In practice, several classical functions have \textsc{robdd}s of small size. For example the symmetrical functions in~$k$ variables are associated with \textsc{robdd}s of quadratic size in~$k$ (see \cite{Knuth11}). Hence, the approach described in this paper leads the way to provide (polynomial) uniform random generator for \textsc{robdd}s of small size (i.e., of size less than exponential). 
%An example of random \textsc{robdd} is depicted in Section \ref{appendix:random} of the appendix.

An interesting future work consists in enumerating the \textsc{bdd} structures where our counting and sampling methods can be applied such as (non-reduced) \textsc{obdd}, or \textsc{zdd} which generally used to represent sets.

Finally, another research direction consists in noting that the generating function of \textsc{robdd}s with both size and number of variables can be specified  thanks to an iterative process as in Theorem~\ref{thm:gf}. It would be interesting to see if the machinery of analytic combinatorics \cite{FlajoletSedgewick2009} is amenable to this kind of specification.

%%%%%%%%%%%%%%%%%%%%%%%%%%%%%%%%%%%%%%%%

%%
%% Bibliography
%%

%\clearpage

\bibliography{ClementGenitrini}
%%%%%%%%%%%%%%%%%%%%%%%%%%%%%%%%%%%%%%%%%%%%%%%%%%%%%%%%%%%%%%%%

%\clearpage 
%
\appendix
\section{Appendix}

%\noteAnt{il faut peut-être placer les 2 exemples suivants en annexe ?}

%% Then we are able to compute the generating function of all \textsc{robdd}s with at most $k$ variables 
%% \begin{equation}
%%   \label{eq:Hk}
%%   H_{k}(u) = \sum_{i\ge 0} N_i u^i = G_{k}(u, 2)\quad \text{where $\displaystyle{G_k(u, X) = \underbrace{\varphi \circ \dots \circ \varphi}_{\text{$k$ times}}(X)}$}.
%% \end{equation}
\subsection{Iterating~$\varphi$ (an example)}
\label{appendix:iterate}
In this section, we illustrate how $\varphi$ can be iterated to count the distribution of \textsc{robdd}s up to $4$  variables.
  Let us  consider the set of Boolean function~$\mathcal{B}_k$ for~$k\in\{1, 2, 3\}$ and pose~$H_k(u, X) = \varphi^k[X]$.
  \begin{itemize}
  \item
    $H_1(u, X) = \varphi[X]=(\phi_0(X)+u \phi_1(X)) = {\left(X^{2} - X\right)} u + X^2$. 
    We can verify that substituting~$X=2$ we get~$F_1(u)=H_1(u, 2) = 2 \, u + 2$. Indeed,~$\mathcal{B}_1 = \{ \bot, x_1, \overline{x_1}, \top\}$, with respective truth tables~$\{\texttt{00}, \texttt{01}, \texttt{10}, \texttt{11}\}$ leading to~$2$ \textsc{robdd}s of size~$1$ and~$2$ \textsc{robdd}s of size~$0$, i.e., with no decision node.
  \item
    Adding a second layer, we get~$H_2(u, X) = \varphi(H_1(u, X)) = \varphi^2[X]$, yielding 
    \begin{align*}
     H_2(u, X) =  {\left(X^{4} - 2 \, X^{3} + X\right)} u^{3} + 2 \, {\left(X^{3} - X^{2}\right)} u^{2} + 2 \, {\left(X^{2} - X\right)} u + X.
     % &={\left(X^{2} - X\right)} {\left(X^{2} - X - 1\right)} u^{3} + {\left({\left(X^{2} - X\right)} {\left(2 \, X + 1\right)} - X^{2} + X\right)} u^{2} + {\left(X^{2} - X\right)} u
    \end{align*}
We get~$F_2(u) = 2 \, u^{3} + 8 \, u^{2} + 4 \, u + 2$, hence there are
$2, 8, 4, 2$ \textsc{robdd}s of respective sizes~$3, 2, 1$ and~$0$ for~$16$ Boolean functions on~$2$ variables.
% given by their truth tables $\{\texttt{0000}, \texttt{0001}, \dots, \texttt{1111}\}$.
  \item
    Iterating with a third layer, we get~$H_3(u, X)= \varphi(H_1(u, X))=\varphi^3[X]$ which has~$34$ terms and gives
    \[
      F_3(u) = 74 \, u^{5} + 88 \, u^{4} + 62 \, u^{3} + 24 \, u^{2} + 6 \, u + 2.
    \]
    Hence, there are respectively~$74, 88, 62, 24, 6$ and~$2$ \textsc{robdd}s of size~$5, 4, 3, 2, 1$ and~$0$.
\item
  Adding a fourth layer yields for~$H_4(u, X)=\varphi^k[X]$ a polynomial with~$134$ terms and 
  \[
    F_4(u) =  11160 \, u^{9} + 23280 \, u^{8} + 17666 \, u^{7} + 8928 \, u^{6} + 3248 \, u^{5} + 960 \, u^{4} + 236 \, u^{3} + 48 \, u^{2} + 8 \, u + 2.
  \]
  There are~$11160$ \textsc{robdd}s with~$4$ variables of size~$9$,~$23280$ \textsc{robdd}s of size~$8$, etc. 
\end{itemize}
Of course, we can check that~$F_k(1) = 2^{2^k}$, which is the number of Boolean functions on~$k$ variables. 

\subsection{Truncating polynomials (an example)}
\label{appendix:truncate}
In this subsection, we illustrate on an example the computational effect of truncating polynomials.

Let us consider that  the generating function~$F_{3}^{\le 3}(u)~$ of \textsc{robdd} with~$k=3$ variables and with less or equal to~$n=3$ decision nodes. The maximal size of a \textsc{robdd} for~$k=3$ variables is~$M_3=7$ ($7$ decision nodes).
Then~$\varphi^3[X]$ is a polynomial with~$34$ terms. We can write (terms who would disappear modulo~$u^{n+1}=u^4$ are grayed)
  \begin{small}
  \begin{align*}
    \varphi^3[X] &=
    \textcolor{gray}{\left(X^{8} - 4 \, X^{7} + 14 \, X^{5} - 6 \, X^{4} - 16 \, X^{3} + 5 \, X^{2} + 6 \, X\right)  u^{7}}\\
      & \textcolor{gray}{+ 4 \, {\left(X^{7} - 2 \, X^{6} - 3 \, X^{5} + 5 \, X^{4} + 4 \, X^{3} - 3 \, X^{2} - 2 \, X\right)} u^{6}} \\
      & \textcolor{gray}{ + {\left(8 \, X^{6} - 12 \, X^{5} - 11 \, X^{4} + 14 \, X^{3} + 4 \, X^{2} - 3 \, X\right)} u^{5}} \\
    & \textcolor{gray}{+ 2 \, {\left(5 \, X^{5} - 6 \, X^{4} - 5 \, X^{3} + 4 \, X^{2} + 2 \, X\right)} u^{4}} \\
    &  + {\left(9 \, X^{4} - 10 \, X^{3} - 2 \, X^{2} + 3 \, X\right)} u^{3}\\
    & + 6 \, {\left(X^{3} - X^{2}\right)} u^{2} \\
    &+ 3 \, {\left(X^{2} - X\right)} u \\
    &    + X
  \end{align*}
  \end{small}
  Truncating modulo~$u^4$ and substituting~$X=2$ yields the polynomial
  \[
  F_{3}^{\le4}(u) = 62 \, u^{3} + 24 \, u^{2} + 6 \, u + 2.
  \]
  The problem of computing all the terms before truncating the result is that there are~$O(\frac{2^{2k}}{k^2})$ terms (since~$M_k\approx 2^k/k$), hence a combinatorial explosion.
  
  In contrast, Algorithm \ref{algo:memo} works in a recursive manner and truncate polynomials along the way so that we have a polynomial number of terms.
  In the following, we describe with some details how to compute~$(\varphi^3(X))_{X=2}$.

%%   Our algorithm uses truncated polynomials of the form $\varphi[X^m]$ modulo $u^5$.
%%   For illustration we compute the first polynomials of this family
%%   \begin{footnotesize}
%%     \begin{align*}
%%    \varphi[X] =  &{\left(X^{2} - X\right)} u + X\\
%% \varphi[X^2] = &{\left(X^{4} - 2 \, X^{3} + X\right)} u^{2}  + {\left(2 \, X^{3} - X^{2} - X\right)} u + X^{2}\\
%% \varphi[X^3] = &{\left(X^{6} - 3 \, X^{5} + 5 \, X^{3} - X^{2} - 2 \, X\right)} u^{3}  \\
%% & \qquad + 3 \, {\left(X^{5} - X^{4} - 2 \, X^{3} + X^{2} + X\right)} u^{2} + {\left(3 \, X^{4} - 2 \, X^{2} - X\right)} u + X^{3}\\
%% \varphi[X^4] =  &{\left(X^{8} - 4 \, X^{7} + 14 \, X^{5} - 6 \, X^{4} - 16 \, X^{3} + 5 \, X^{2} + 6 \, X\right)} u^{4} \\
%% \qquad & + 2 \, {\left(2 \, X^{7} - 3 \, X^{6} - 9 \, X^{5} + 10 \, X^{4} + 13 \, X^{3} - 7 \, X^{2} - 6 \, X\right)} u^{3} \\
%% & \qquad \quad+ {\left(6 \, X^{6} - 17 \, X^{4} - 8 \, X^{3} + 12 \, X^{2} + 7 \, X\right)} u^{2} \\
%% &\qquad\quad\quad + {\left(4 \, X^{5} + 2 \, X^{4} - 2 \, X^{3} - 3 \, X^{2} - X\right)} u + X^{4}\\
%%   \end{align*}
%%   \end{footnotesize}
    
  First, Algorithm \ref{algo:memo} decomposes $(\varphi^3(X))_{X=2}$ as~$(\varphi^2\circ \varphi[X]))_{X=2}$.
  In general,~$\varphi[X^m]$ is a polynomial of respective degree~$m$ and~$2m$ in variables~$u$ and~$X$ and has~$O(m^2)$ integer coefficients.
  We compute (collecting terms with respect to variable~$X$)
  \begin{equation}
    \label{eq:collectX}
  \varphi[X] = (X^2-X) u +X = X^2u +X (1-u).
  \end{equation}
Denoting~$B^{(2)}_m(u) =  (\varphi^2[X^m])_{X=2}$, our algorithm computes recursively for the basis~$(1, X, X^2)$ (removed monomials modulo~$u^4$ ar grayed)
  \begin{align*}
    B^{(2)}_0(u) &= 1 \\
    B^{(2)}_1(u) &= 2 \, u^{3} + 8 \, u^{2} + 4 \, u + 2 \\
    B^{(2)}_2(u) &= \textcolor{gray}{74 \, u^{4}} + 90 \, u^{3} + 68 \, u^{2} + 20 \, u + 4
  \end{align*}
  Then since~$\varphi$ is linear, we compute (still modulo~$u^4$) using Equation \eqref{eq:collectX}
  \begin{align*}
    (\varphi^3(X))_{X=2}
    &= B_2(u) u + B_1(u) (1-u)\\
    &= u (90 \, u^{3} + 68 \, u^{2} + 20 \, u + 4) +(1-u) (2 \, u^{3} + 8 \, u^{2} + 4 \, u + 2 )\\
    & = \textcolor{gray}{88 \, u^{4}} + 62 \, u^{3} + 24 \, u^{2} + 6 \, u + 2.
    \end{align*}
  In summary, computing modulo~$u^{n+1}$ along the process allows us to control the degree~$O(n)$ of polynomials involved, which in turn ensures that the computation stays of polynomial complexity (with respect ti arithmetical operations on integers).

Observing such curves for~$k$ from~$1$ to~$13$, we notice the exponential growth of the largest \textsc{robdd}s
when the number~$k$ of variables increases. Indeed, in Theorem~\ref{theo:largest_robdd}
we define~$M_k$ to be the size of the largest \textsc{robdd}s with~$k$ variables.
The sequence starts as
$(M_k)_{k=1,\dots, 13} = (1, 3, 5, 9, 17, 29, 45, 77, 141, 269, 509, 765, 1277).$ 

%\subsection{An example of random ROBDD}
%\label{appendix:random}
%The cardinality of the set \textsc{robdd}s of size~$60$ with at most~$15$ variables (maximal size is~$M_{15} = 4349$) is
%approximately~$1.59\cdot 10^{114}$.
%The \textsc{robdd} with profile~$\boldsymbol{p}=\SEQ{1, 2, 3, 4, 2, 4, 4, 6, 4, 7, 6, 6, 4, 3, 2}$ depicted on Figure~\ref{fig:big-robdd} was produced by a random uniform sampler derived from the counting method presented in this paper.
%\begin{figure}[htbp]
%  %\centering{\includegraphics[width=0.77\textwidth]{./figs/robdd-15-60.tex}}
%  \centering{\scalebox{.35}{\input{./figs/robdd-15-60.tex}}}
%  \caption{\label{fig:big-robdd}An example of \textsc{robdd} of size~$60$ with~$15$ variables.}
%\end{figure}

\end{document}

%% file: figs/fig-black-6_15_2289425208767.tikz
\begin{tikzpicture}[>=latex,line join=bevel,]
  \pgfsetlinewidth{1bp}
\begin{scope}
  \pgfsetstrokecolor{black}
  \definecolor{strokecol}{rgb}{1.0,1.0,1.0};
  \pgfsetstrokecolor{strokecol}
  \definecolor{fillcol}{rgb}{1.0,1.0,1.0};
  \pgfsetfillcolor{fillcol}
  \filldraw (0.0bp,0.0bp) -- (0.0bp,300.0bp) -- (320.0bp,300.0bp) -- (320.0bp,0.0bp) -- cycle;
\end{scope}
\begin{scope}
  \pgfsetstrokecolor{black}
  \definecolor{strokecol}{rgb}{1.0,1.0,1.0};
  \pgfsetstrokecolor{strokecol}
  \definecolor{fillcol}{rgb}{1.0,1.0,1.0};
  \pgfsetfillcolor{fillcol}
  \filldraw (0.0bp,0.0bp) -- (0.0bp,300.0bp) -- (320.0bp,300.0bp) -- (320.0bp,0.0bp) -- cycle;
\end{scope}
\begin{scope}
  \pgfsetstrokecolor{black}
  \definecolor{strokecol}{rgb}{1.0,1.0,1.0};
  \pgfsetstrokecolor{strokecol}
  \definecolor{fillcol}{rgb}{1.0,1.0,1.0};
  \pgfsetfillcolor{fillcol}
  \filldraw (0.0bp,0.0bp) -- (0.0bp,300.0bp) -- (320.0bp,300.0bp) -- (320.0bp,0.0bp) -- cycle;
\end{scope}
\begin{scope}
  \pgfsetstrokecolor{black}
  \definecolor{strokecol}{rgb}{1.0,1.0,1.0};
  \pgfsetstrokecolor{strokecol}
  \definecolor{fillcol}{rgb}{1.0,1.0,1.0};
  \pgfsetfillcolor{fillcol}
  \filldraw (0.0bp,0.0bp) -- (0.0bp,300.0bp) -- (320.0bp,300.0bp) -- (320.0bp,0.0bp) -- cycle;
\end{scope}
\begin{scope}
  \pgfsetstrokecolor{black}
  \definecolor{strokecol}{rgb}{1.0,1.0,1.0};
  \pgfsetstrokecolor{strokecol}
  \definecolor{fillcol}{rgb}{1.0,1.0,1.0};
  \pgfsetfillcolor{fillcol}
  \filldraw (0.0bp,0.0bp) -- (0.0bp,300.0bp) -- (320.0bp,300.0bp) -- (320.0bp,0.0bp) -- cycle;
\end{scope}
\begin{scope}
  \pgfsetstrokecolor{black}
  \definecolor{strokecol}{rgb}{1.0,1.0,1.0};
  \pgfsetstrokecolor{strokecol}
  \definecolor{fillcol}{rgb}{1.0,1.0,1.0};
  \pgfsetfillcolor{fillcol}
  \filldraw (0.0bp,0.0bp) -- (0.0bp,300.0bp) -- (320.0bp,300.0bp) -- (320.0bp,0.0bp) -- cycle;
\end{scope}
\begin{scope}
  \pgfsetstrokecolor{black}
  \definecolor{strokecol}{rgb}{1.0,1.0,1.0};
  \pgfsetstrokecolor{strokecol}
  \definecolor{fillcol}{rgb}{1.0,1.0,1.0};
  \pgfsetfillcolor{fillcol}
  \filldraw (0.0bp,0.0bp) -- (0.0bp,300.0bp) -- (320.0bp,300.0bp) -- (320.0bp,0.0bp) -- cycle;
\end{scope}
\begin{scope}
  \pgfsetstrokecolor{black}
  \definecolor{strokecol}{rgb}{1.0,1.0,1.0};
  \pgfsetstrokecolor{strokecol}
  \definecolor{fillcol}{rgb}{1.0,1.0,1.0};
  \pgfsetfillcolor{fillcol}
  \filldraw (0.0bp,0.0bp) -- (0.0bp,300.0bp) -- (320.0bp,300.0bp) -- (320.0bp,0.0bp) -- cycle;
\end{scope}
  \pgfsetcolor{black}
  % Edge: t0 -> t5
  \draw [->,dotted] (114.0bp,283.0bp) .. controls (108.12bp,277.12bp) and (108.03bp,267.82bp)  .. (111.22bp,252.61bp);
  % Edge: t0 -> t1
  \draw [->] (128.0bp,283.0bp) .. controls (135.36bp,275.64bp) and (142.32bp,266.22bp)  .. (151.62bp,252.22bp);
  % Edge: t5 -> t9
  \draw [->,dotted] (107.0bp,236.0bp) .. controls (100.13bp,229.13bp) and (94.267bp,219.96bp)  .. (86.417bp,205.39bp);
  % Edge: t5 -> t6
  \draw [->] (121.0bp,236.0bp) .. controls (126.88bp,230.12bp) and (126.97bp,220.82bp)  .. (123.78bp,205.61bp);
  % Edge: t9 -> t11
  \draw [->,dotted] (76.0bp,189.0bp) .. controls (56.079bp,169.08bp) and (57.147bp,151.56bp)  .. (69.0bp,126.0bp) .. controls (71.474bp,120.67bp) and (75.949bp,116.16bp)  .. (86.988bp,108.2bp);
  % Edge: t9 -> t10
  \draw [->] (90.0bp,189.0bp) .. controls (109.84bp,169.16bp) and (109.56bp,159.1bp)  .. (131.0bp,141.0bp) .. controls (145.73bp,128.57bp) and (165.14bp,117.72bp)  .. (185.23bp,107.6bp);
  % Edge: t11 -> t14
  \draw [->,dotted] (88.0bp,95.0bp) .. controls (59.511bp,66.511bp) and (122.82bp,32.008bp)  .. (161.36bp,14.103bp);
  % Edge: t11 -> t12
  \draw [->] (102.0bp,95.0bp) .. controls (125.58bp,71.421bp) and (165.11bp,62.293bp)  .. (195.04bp,57.858bp);
  % Edge: t12 -> t14
  \draw [->,dotted] (197.0bp,48.0bp) .. controls (190.22bp,41.217bp) and (184.08bp,32.444bp)  .. (175.41bp,18.208bp);
  % Edge: t12 -> t13
  \draw [->] (211.0bp,48.0bp) .. controls (221.14bp,37.857bp) and (234.0bp,27.977bp)  .. (250.4bp,16.359bp);
  % Edge: t10 -> t12
  \draw [->] (200.0bp,95.0bp) .. controls (205.72bp,89.282bp) and (206.9bp,80.393bp)  .. (205.75bp,64.9bp);
  % Edge: t10 -> t13
  \draw [->,dotted] (186.0bp,95.0bp) .. controls (170.9bp,79.902bp) and (176.5bp,64.989bp)  .. (188.0bp,47.0bp) .. controls (200.23bp,27.87bp) and (225.7bp,18.149bp)  .. (250.47bp,11.733bp);
  % Edge: t6 -> t11
  \draw [->] (128.0bp,189.0bp) .. controls (138.41bp,178.59bp) and (116.28bp,138.45bp)  .. (99.438bp,110.88bp);
  % Edge: t6 -> t7
  \draw [->,dotted] (114.0bp,189.0bp) .. controls (104.1bp,179.1bp) and (119.54bp,166.69bp)  .. (139.21bp,154.93bp);
  % Edge: t7 -> t14
  \draw [->,dotted] (140.0bp,142.0bp) .. controls (128.06bp,130.06bp) and (140.0bp,120.89bp)  .. (140.0bp,104.0bp) .. controls (140.0bp,104.0bp) and (140.0bp,104.0bp)  .. (140.0bp,55.0bp) .. controls (140.0bp,42.904bp) and (147.69bp,31.695bp)  .. (161.19bp,17.835bp);
  % Edge: t7 -> t8
  \draw [->] (154.0bp,142.0bp) .. controls (179.69bp,116.31bp) and (223.55bp,107.96bp)  .. (254.93bp,104.53bp);
  % Edge: t8 -> t14
  \draw [->,dotted] (257.0bp,95.0bp) .. controls (237.95bp,75.954bp) and (239.14bp,65.951bp)  .. (220.0bp,47.0bp) .. controls (209.91bp,37.006bp) and (197.04bp,27.392bp)  .. (180.62bp,16.135bp);
  % Edge: t8 -> t13
  \draw [->] (271.0bp,95.0bp) .. controls (289.3bp,76.701bp) and (277.66bp,44.036bp)  .. (264.79bp,18.161bp);
  % Edge: t1 -> t2
  \draw [->,dotted] (149.0bp,236.0bp) .. controls (142.16bp,229.16bp) and (146.71bp,218.71bp)  .. (157.19bp,204.25bp);
  % Edge: t1 -> t4
  \draw [->] (163.0bp,236.0bp) .. controls (182.97bp,216.03bp) and (195.43bp,184.89bp)  .. (203.8bp,159.04bp);
  % Edge: t2 -> t4
  \draw [->,dotted] (156.0bp,189.0bp) .. controls (142.45bp,175.45bp) and (171.08bp,162.59bp)  .. (197.66bp,153.45bp);
  % Edge: t2 -> t3
  \draw [->] (170.0bp,189.0bp) .. controls (174.22bp,184.78bp) and (204.97bp,168.3bp)  .. (230.35bp,154.97bp);
  % Edge: t4 -> t13
  \draw [->] (213.0bp,142.0bp) .. controls (248.68bp,106.32bp) and (320.0bp,154.47bp)  .. (320.0bp,104.0bp) .. controls (320.0bp,104.0bp) and (320.0bp,104.0bp)  .. (320.0bp,55.0bp) .. controls (320.0bp,33.433bp) and (294.88bp,20.936bp)  .. (269.64bp,12.42bp);
  % Edge: t4 -> t10
  \draw [->,dotted] (199.0bp,142.0bp) .. controls (193.31bp,136.31bp) and (191.65bp,127.57bp)  .. (191.77bp,112.19bp);
  % Edge: t3 -> t12
  \draw [->,dotted] (231.0bp,142.0bp) .. controls (210.72bp,121.72bp) and (227.54bp,105.68bp)  .. (217.0bp,79.0bp) .. controls (215.89bp,76.194bp) and (214.45bp,73.327bp)  .. (208.94bp,64.039bp);
  % Edge: t3 -> t10
  \draw [->] (245.0bp,142.0bp) .. controls (258.96bp,128.04bp) and (228.93bp,115.26bp)  .. (201.65bp,106.39bp);
  % Node: t0
\begin{scope}
  \definecolor{strokecol}{rgb}{0.0,0.0,0.0};
  \pgfsetstrokecolor{strokecol}
  \draw (121.0bp,291.0bp) ellipse (9.0bp and 9.0bp);
  \draw (121.0bp,291.0bp) node[font=\large] {$x_1$};
\end{scope}
  % Node: t5
\begin{scope}
  \definecolor{strokecol}{rgb}{0.0,0.0,0.0};
  \pgfsetstrokecolor{strokecol}
  \draw (114.0bp,244.0bp) ellipse (9.0bp and 9.0bp);
  \draw (114.0bp,244.0bp) node[font=\large] {$x_2$};
\end{scope}
  % Node: t1
\begin{scope}
  \definecolor{strokecol}{rgb}{0.0,0.0,0.0};
  \pgfsetstrokecolor{strokecol}
  \draw (156.0bp,244.0bp) ellipse (9.0bp and 9.0bp);
  \draw (156.0bp,244.0bp) node[font=\large] {$x_2$};
\end{scope}
  % Node: t9
\begin{scope}
  \definecolor{strokecol}{rgb}{0.0,0.0,0.0};
  \pgfsetstrokecolor{strokecol}
  \draw (83.0bp,197.0bp) ellipse (9.0bp and 9.0bp);
  \draw (83.0bp,197.0bp) node[font=\large] {$x_3$};
\end{scope}
  % Node: t6
\begin{scope}
  \definecolor{strokecol}{rgb}{0.0,0.0,0.0};
  \pgfsetstrokecolor{strokecol}
  \draw (121.0bp,197.0bp) ellipse (9.0bp and 9.0bp);
  \draw (121.0bp,197.0bp) node[font=\large] {$x_3$};
\end{scope}
  % Node: t11
\begin{scope}
  \definecolor{strokecol}{rgb}{0.0,0.0,0.0};
  \pgfsetstrokecolor{strokecol}
  \draw (95.0bp,103.0bp) ellipse (9.0bp and 9.0bp);
  \draw (95.0bp,103.0bp) node[font=\large] {$x_5$};
\end{scope}
  % Node: t10
\begin{scope}
  \definecolor{strokecol}{rgb}{0.0,0.0,0.0};
  \pgfsetstrokecolor{strokecol}
  \draw (193.0bp,103.0bp) ellipse (9.0bp and 9.0bp);
  \draw (193.0bp,103.0bp) node[font=\large] {$x_5$};
\end{scope}
  % Node: t14
\begin{scope}
  \definecolor{strokecol}{rgb}{0.0,0.0,0.0};
  \pgfsetstrokecolor{strokecol}
  \draw (180.5bp,18.0bp) -- (161.5bp,18.0bp) -- (161.5bp,0.0bp) -- (180.5bp,0.0bp) -- cycle;
  \draw (171.0bp,9.0bp) node[font=\large] {$\bot$};
\end{scope}
  % Node: t12
\begin{scope}
  \definecolor{strokecol}{rgb}{0.0,0.0,0.0};
  \pgfsetstrokecolor{strokecol}
  \draw (204.0bp,56.0bp) ellipse (9.0bp and 9.0bp);
  \draw (204.0bp,56.0bp) node[font=\large] {$x_6$};
\end{scope}
  % Node: t13
\begin{scope}
  \definecolor{strokecol}{rgb}{0.0,0.0,0.0};
  \pgfsetstrokecolor{strokecol}
  \draw (269.5bp,18.0bp) -- (250.5bp,18.0bp) -- (250.5bp,0.0bp) -- (269.5bp,0.0bp) -- cycle;
  \draw (260.0bp,9.0bp) node[font=\large] {$\top$};
\end{scope}
  % Node: t7
\begin{scope}
  \definecolor{strokecol}{rgb}{0.0,0.0,0.0};
  \pgfsetstrokecolor{strokecol}
  \draw (147.0bp,150.0bp) ellipse (9.0bp and 9.0bp);
  \draw (147.0bp,150.0bp) node[font=\large] {$x_4$};
\end{scope}
  % Node: t8
\begin{scope}
  \definecolor{strokecol}{rgb}{0.0,0.0,0.0};
  \pgfsetstrokecolor{strokecol}
  \draw (264.0bp,103.0bp) ellipse (9.0bp and 9.0bp);
  \draw (264.0bp,103.0bp) node[font=\large] {$x_5$};
\end{scope}
  % Node: t2
\begin{scope}
  \definecolor{strokecol}{rgb}{0.0,0.0,0.0};
  \pgfsetstrokecolor{strokecol}
  \draw (163.0bp,197.0bp) ellipse (9.0bp and 9.0bp);
  \draw (163.0bp,197.0bp) node[font=\large] {$x_3$};
\end{scope}
  % Node: t4
\begin{scope}
  \definecolor{strokecol}{rgb}{0.0,0.0,0.0};
  \pgfsetstrokecolor{strokecol}
  \draw (206.0bp,150.0bp) ellipse (9.0bp and 9.0bp);
  \draw (206.0bp,150.0bp) node[font=\large] {$x_4$};
\end{scope}
  % Node: t3
\begin{scope}
  \definecolor{strokecol}{rgb}{0.0,0.0,0.0};
  \pgfsetstrokecolor{strokecol}
  \draw (238.0bp,150.0bp) ellipse (9.0bp and 9.0bp);
  \draw (238.0bp,150.0bp) node[font=\large] {$x_4$};
\end{scope}
\end{tikzpicture}

%% file: figs/fig-black-spine_6_15_2289425208767.tikz
\begin{tikzpicture}[>=latex,line join=bevel,]
  \pgfsetlinewidth{1bp}
%%

  % Node: t0
  % top of the tree
\pgfsetstrokecolor{black!35}
\begin{scope}
 \draw (121.0bp,291.0bp) ellipse (9.0bp and 9.0bp);
  \draw (121.0bp,291.0bp) node[font=\large,color=black!35] {$x_1$};
\end{scope}
  % Node: t1
\begin{scope}
  \draw (156.0bp,244.0bp) ellipse (9.0bp and 9.0bp);
  \draw (156.0bp,244.0bp) node[font=\large, color=black!35] {$x_2$};
\end{scope}
   % Node: t2
\begin{scope}
  \draw (163.0bp,197.0bp) ellipse (9.0bp and 9.0bp);
  \draw (163.0bp,197.0bp) node[color=black!35, font=\large] {$x_3$};
\end{scope}
% Node: t5
\begin{scope}
  \draw (114.0bp,244.0bp) ellipse (9.0bp and 9.0bp);
  \draw (114.0bp,244.0bp) node[font=\large, color=black!35] {$x_2$};
\end{scope}
  % Node: t6
\begin{scope}
  \draw (121.0bp,197.0bp) ellipse (9.0bp and 9.0bp);
  \draw (121.0bp,197.0bp) node[font=\large, color=black!35] {$x_3$};
\end{scope}
  % Node: t9
\begin{scope}
  \draw (83.0bp,197.0bp) ellipse (9.0bp and 9.0bp);
  \draw (83.0bp,197.0bp) node[font=\large, color=black!35] {$x_3$};
\end{scope}

  \pgfsetcolor{black!35}
   % Edge: t0 -> t5
   \draw [->,dotted] (114.0bp,283.0bp) .. controls (108.12bp,277.12bp) and (108.03bp,267.82bp)  .. (111.22bp,252.61bp);
   % Edge: t0 -> t1
   \draw [->] (128.0bp,283.0bp) .. controls (135.36bp,275.64bp) and (142.32bp,266.22bp)  .. (151.62bp,252.22bp);
   % Edge: t5 -> t9
   \draw [->,dotted] (107.0bp,236.0bp) .. controls (100.13bp,229.13bp) and (94.267bp,219.96bp)  .. (86.417bp,205.39bp);
   % Edge: t5 -> t6
   \draw [->] (121.0bp,236.0bp) .. controls (126.88bp,230.12bp) and (126.97bp,220.82bp)  .. (123.78bp,205.61bp);
  % Edge: t1 -> t2
   \draw [->,dotted] (149.0bp,236.0bp) .. controls (142.16bp,229.16bp) and (146.71bp,218.71bp)  .. (157.19bp,204.25bp);

\pgfsetcolor{red!35}
% Edge: t1 -> t4
\pgfsetcolor{red}
\draw [->]  (192.8bp,169.04bp) --  (202.8bp,159.04bp);
\pgfsetcolor{red!35}
\draw [->] (163.0bp,236.0bp) -- (203.8bp,175bp);
% Edge: t2 -> t4
\pgfsetcolor{red}
\draw [->] (187.66bp,163.45bp) -- (197.66bp,153.45bp);
\pgfsetcolor{red!35}
\draw [->,dotted] (156.0bp,189.0bp) -- (149.66bp,175bp);
% Edge: t9 -> t11		     
\pgfsetcolor{red}
\draw [->] (76.0bp,118.0bp) -- (86.988bp,108.2bp);
\pgfsetcolor{red!35}
\draw [->,dotted] (76.0bp,189.0bp) -- (69.0bp,175.0bp);
% Edge: t9 -> t10
\pgfsetcolor{red}
\draw [->]  (175.23bp,117.6bp) -- (185.23bp,107.6bp);
\pgfsetcolor{red!35}
\draw [->] (90.0bp,189.0bp) -- (97.0bp,175.0bp); 
% Edge: t6 -> t11
\pgfsetcolor{red}
\draw [->] (82.438bp,121.88bp) -- (92.438bp,111.88bp);
\pgfsetcolor{red!35}
\draw [->] (128.0bp,189.0bp) -- (135.438bp,175.0bp);
% Edge: t2 -> t3
\pgfsetcolor{red}
\draw [->] (220.35bp,164.97bp) -- (230.35bp,154.97bp);
\pgfsetcolor{red!35}
\draw [->] (170.0bp,189.0bp) -- (177.35bp,175bp);
% Edge: t6 -> t7
\pgfsetcolor{red}
\draw [->] (129.21bp,164.93bp) -- (139.21bp,154.93bp);
\pgfsetcolor{red!35}
\draw [->,dotted] (114.0bp,189.0bp) -- (107.21bp,175.0bp);

\pgfsetcolor{black}
\draw[dashed]  (53.0bp,172.0bp) -- (300bp,172bp);

\pgfsetcolor{black}

% Edge: t12 -> t13
  \draw [->] (211.0bp,48.0bp) .. controls (221.14bp,37.857bp) and (234.0bp,27.977bp)  .. (250.4bp,16.359bp);
  % Edge: t10 -> t12
  \draw [->] (200.0bp,95.0bp) .. controls (205.72bp,89.282bp) and (206.9bp,80.393bp)  .. (205.75bp,64.9bp);
  % Edge: t10 -> t13
  \draw [->,dotted] (186.0bp,95.0bp) .. controls (170.9bp,79.902bp) and (176.5bp,64.989bp)  .. (188.0bp,47.0bp) .. controls (200.23bp,27.87bp) and (225.7bp,18.149bp)  .. (250.47bp,11.733bp);

% Edge: t11 -> t14
  \draw [->,dotted] (88.0bp,95.0bp) .. controls (59.511bp,66.511bp) and (122.82bp,32.008bp)  .. (161.36bp,14.103bp);
  % Edge: t11 -> t12
  \draw [->] (102.0bp,95.0bp) .. controls (125.58bp,71.421bp) and (165.11bp,62.293bp)  .. (195.04bp,57.858bp);
  % Edge: t12 -> t14
  \draw [->,dotted] (197.0bp,48.0bp) .. controls (190.22bp,41.217bp) and (184.08bp,32.444bp)  .. (175.41bp,18.208bp);
% Edge: t7 -> t14
 \draw [->,dotted] (140.0bp,142.0bp) .. controls (128.06bp,130.06bp) and (140.0bp,120.89bp)  .. (140.0bp,104.0bp) .. controls (140.0bp,104.0bp) and (140.0bp,104.0bp)  .. (140.0bp,55.0bp) .. controls (140.0bp,42.904bp) and (147.69bp,31.695bp)  .. (161.19bp,17.835bp);
  % Edge: t7 -> t8
  \draw [->] (154.0bp,142.0bp) .. controls (179.69bp,116.31bp) and (223.55bp,107.96bp)  .. (254.93bp,104.53bp);
  % Edge: t8 -> t14
  \draw [->,dotted] (257.0bp,95.0bp) .. controls (237.95bp,75.954bp) and (239.14bp,65.951bp)  .. (220.0bp,47.0bp) .. controls (209.91bp,37.006bp) and (197.04bp,27.392bp)  .. (180.62bp,16.135bp);
  % Edge: t8 -> t13
  \draw [->] (271.0bp,95.0bp) .. controls (289.3bp,76.701bp) and (277.66bp,44.036bp)  .. (264.79bp,18.161bp);
  % Edge: t4 -> t13
  \draw [->] (213.0bp,142.0bp) .. controls (248.68bp,106.32bp) and (320.0bp,154.47bp)  .. (320.0bp,104.0bp) .. controls (320.0bp,104.0bp) and (320.0bp,104.0bp)  .. (320.0bp,55.0bp) .. controls (320.0bp,33.433bp) and (294.88bp,20.936bp)  .. (269.64bp,12.42bp);
  % Edge: t4 -> t10
  \draw [->,dotted] (199.0bp,142.0bp) .. controls (193.31bp,136.31bp) and (191.65bp,127.57bp)  .. (191.77bp,112.19bp);
  % Edge: t3 -> t12
  \draw [->,dotted] (231.0bp,142.0bp) .. controls (210.72bp,121.72bp) and (227.54bp,105.68bp)  .. (217.0bp,79.0bp) .. controls (215.89bp,76.194bp) and (214.45bp,73.327bp)  .. (208.94bp,64.039bp);
  % Edge: t3 -> t10
  \draw [->] (245.0bp,142.0bp) .. controls (258.96bp,128.04bp) and (228.93bp,115.26bp)  .. (201.65bp,106.39bp);
%  Node: t11
\begin{scope}
  \definecolor{strokecol}{rgb}{0.0,0.0,0.0};
  \pgfsetstrokecolor{strokecol}
  \draw (95.0bp,103.0bp) ellipse (9.0bp and 9.0bp);
  \draw (95.0bp,103.0bp) node[font=\large] {$x_5$};
\end{scope}
  % Node: t10
\begin{scope}
  \definecolor{strokecol}{rgb}{0.0,0.0,0.0};
  \pgfsetstrokecolor{strokecol}
  \draw (193.0bp,103.0bp) ellipse (9.0bp and 9.0bp);
  \draw (193.0bp,103.0bp) node[font=\large] {$x_5$};
\end{scope}
  % Node: t14
\begin{scope}
  \definecolor{strokecol}{rgb}{0.0,0.0,0.0};
  \pgfsetstrokecolor{strokecol}
  \draw (180.5bp,18.0bp) -- (161.5bp,18.0bp) -- (161.5bp,0.0bp) -- (180.5bp,0.0bp) -- cycle;
  \draw (171.0bp,9.0bp) node[font=\large] {$\bot$};
\end{scope}
  % Node: t12
\begin{scope}
  \definecolor{strokecol}{rgb}{0.0,0.0,0.0};
  \pgfsetstrokecolor{strokecol}
  \draw (204.0bp,56.0bp) ellipse (9.0bp and 9.0bp);
  \draw (204.0bp,56.0bp) node[font=\large] {$x_6$};
\end{scope}
  % Node: t13
\begin{scope}
  \definecolor{strokecol}{rgb}{0.0,0.0,0.0};
  \pgfsetstrokecolor{strokecol}
  \draw (269.5bp,18.0bp) -- (250.5bp,18.0bp) -- (250.5bp,0.0bp) -- (269.5bp,0.0bp) -- cycle;
  \draw (260.0bp,9.0bp) node[font=\large] {$\top$};
\end{scope}
  % Node: t7
\begin{scope}
  \definecolor{strokecol}{rgb}{0.0,0.0,0.0};
  \pgfsetstrokecolor{strokecol}
  \draw (147.0bp,150.0bp) ellipse (9.0bp and 9.0bp);
  \draw (147.0bp,150.0bp) node[font=\large] {$x_4$};
\end{scope}
  % Node: t8
\begin{scope}
  \definecolor{strokecol}{rgb}{0.0,0.0,0.0};
  \pgfsetstrokecolor{strokecol}
  \draw (264.0bp,103.0bp) ellipse (9.0bp and 9.0bp);
  \draw (264.0bp,103.0bp) node[font=\large] {$x_5$};
\end{scope}
  % Node: t4
\begin{scope}
  \definecolor{strokecol}{rgb}{0.0,0.0,0.0};
  \pgfsetstrokecolor{strokecol}
  \draw (206.0bp,150.0bp) ellipse (9.0bp and 9.0bp);
  \draw (206.0bp,150.0bp) node[font=\large] {$x_4$};
\end{scope}
  % Node: t3
\begin{scope}
  \definecolor{strokecol}{rgb}{0.0,0.0,0.0};
  \pgfsetstrokecolor{strokecol}
  \draw (238.0bp,150.0bp) ellipse (9.0bp and 9.0bp);
  \draw (238.0bp,150.0bp) node[font=\large] {$x_4$};
\end{scope}
\end{tikzpicture}